\documentclass[jair,twoside,11pt,theapa]{article}
\usepackage{jair, theapa, rawfonts}

\jairheading{?}{?}{?}{?}{?}
\ShortHeadings{Ordinal Maximin Share Approximation for Goods}
{Hosseini, Searns, Segal-Halevi}
\firstpageno{1}

\usepackage{amsmath,amsthm,amssymb}
\usepackage{booktabs} 
\usepackage[ruled]{algorithm2e} 
\usepackage{algorithmic}

\SetAlFnt{\small}
\SetAlCapFnt{\small}
\SetAlCapNameFnt{\small}
\SetAlCapHSkip{0pt}
\IncMargin{-\parindent}

\usepackage{tabularx}
\usepackage{tikz}
\usetikzlibrary{calc}

\usepackage{subcaption}
\usepackage{url}

\graphicspath{{Experiments/}}

\theoremstyle{plain}
\newtheorem{theorem}{Theorem}
\newtheorem{lemma}{Lemma}
\newtheorem{proposition}{Proposition}

\theoremstyle{definition}
\newtheorem{remark}{Remark}
\newtheorem{example}{Example}
\newtheorem{definition}{Definition}

\newcommand{\partition}[2]{\textsc{Partitions}(#1,#2)}
\newcommand{\union}[2]{\textsc{Union}(#1,#2)}
\newcommand{\MMS}{\textsc{MMS}}

\newcommand{\mms}[4]{\MMS^{#2\text{-out-of-}#3}_#1\left(
#4
\right)}

\newcommand{\MMSA}[2]{\MMS^{#1\text{-out-of-}#2}}

\newcommand{\BBFS}{\textsc{BBFS}}
\newcommand{\BBFSA}[1]{\BBFS^{#1}}
\newcommand{\JSS}{\textsc{JSS}}
\newcommand{\JSSA}[1]{\JSS^{#1}}

\newcommand{\itemrect}[5]{
\draw (#1,#2) rectangle +(#3,#4) node[pos=.5] {#5};
}

\newcommand{\mmsbundle}[5]{
\itemrect{#1}{0} {2}{#2} {#3}
\draw (#1,0) rectangle +(2,#2) rectangle +(0,#4) node[pos=.5] {#5};
}

\newcommand{\ins}[1]{\langle #1 \rangle}

\newcommand{\range}[2]{\in\{#1,\dots,#2\}}

\newcommand{\HiddenExplanation}[1]{}

\newcommand{\ceil}[1]{\lceil #1 \rceil}
\newcommand{\floor}[1]{\lfloor #1 \rfloor}

\newcommand*\circled[1]{\tikz[baseline=(char.base)]{
            \node[shape=circle,draw,inner sep=2pt] (char) {#1};}}

\def\DRAFT{}

\ifdefined\DRAFT

\newcommand{\HH}[1]{{\color{magenta}{HH: }{#1} \color{red}}}

\definecolor{ForestGreen}{rgb}{.13,.54,.13}
\newcommand{\er}[1]{\textcolor{ForestGreen}{#1}}
\newcommand{\erel}[1]{\textcolor{ForestGreen}{({Erel:} #1)}}

\newcommand{\AS}[1]{{\color{teal}{AS: }{#1} \color{green}}}

\else  

\newcommand{\HH}[1]{}

\newcommand{\erel}[1]{}
\newcommand{\er}[1]{#1}

\newcommand{\AS}[1]{}

\fi

\newcommand{\goodsapp}{\floor{ (\ell+\frac{1}{2})n}}

\newcommand{\citet}[1]{\citeA{#1}}
\newcommand{\citep}[1]{\cite{#1}}

\begin{document}

\title{Ordinal Maximin Share Approximation for Goods}

\author{\name Hadi Hosseini 
\email hadi@psu.edu
\\
\addr Pennsylvania State University, University Park
\\
\\
\name Andrew Searns 
\email andrew.searns@jhuapl.edu
\\
\addr Johns Hopkins University Applied Physics Laboratory
\\
\AND
\name Erel Segal-Halevi 
\email erelsgl@gmail.com
\\
\addr Ariel University, Ariel 40700
}

\maketitle

\begin{abstract}
In fair division of indivisible goods, $\ell$-out-of-$d$ maximin share (MMS) is the value that an agent can guarantee by partitioning the goods into $d$ bundles and choosing the $\ell$ least preferred bundles. 
Most existing works aim to guarantee to all agents a constant fraction of their 1-out-of-$n$ MMS. But this guarantee is sensitive to small perturbation in agents' cardinal valuations. We consider a more robust approximation notion, which depends only on the agents' \emph{ordinal} rankings of bundles. 
We prove the existence of $\ell$-out-of-$\floor{(\ell+\frac{1}{2})n}$ MMS allocations of goods for any integer $\ell\geq 1$, and present a polynomial-time algorithm that finds a $1$-out-of-$\ceil{\frac{3n}{2}}$ MMS allocation when $\ell = 1$. We further develop an algorithm that provides a weaker ordinal approximation to MMS for any $\ell > 1$.
\end{abstract}

\section{Introduction}
Fair division is the study of how to distribute a set of items among a set of agents in a fair manner. Achieving fairness is particularly challenging when items are \emph{in}divisible. Computational and conceptual challenges have motivated researchers and practitioners to develop a variety of fairness concepts that are applicable to a large number of allocation problems.\footnote{See \citet{Bouveret2016Fair,lang2016fair,markakis2017approximation,doi:10.1146/annurev-economics-080218-025559} for detailed surveys and discussions.} One of the most common fairness concepts, proposed by \citet{budish2011combinatorial}, is Maximin Share (MMS), that aims to give each agent a bundle that is valued at a certain threshold. The MMS threshold, also known as $1$-out-of-$d$ MMS, generalizes the guarantee of the \textit{cut-and-choose} protocol.
It is the value that an agent can secure by partitioning the items into $d$ bundles, assuming it will receive the least preferred bundle.
The MMS value depends on the number of partitions, $d$. When all items are goods (i.e., have non-negative values), the $1$-out-of-$d$ MMS threshold is (weakly) monotonically decreasing as the number of partitions ($d$) increases. 

When allocating goods among $n$ agents, a natural desirable threshold is satisfying $1$-out-of-$n$ MMS for all agents. Unfortunately, while this value can be guaranteed for $n=2$ agents through the cut-and-choose protocol, a $1$-out-of-$n$ MMS allocation of goods may not exist in general for $n \geq 3$ \citep{procaccia2014fair,kurokawa2018fair}.
These negative results have given rise to \emph{multiplicative approximations}, wherein each agent is guaranteed at least a constant fraction of its 1-out-of-$n$ MMS. While there have been many attempts in developing algorithms that improve the bound to close to 1, the best currently known fraction is $\frac{3}{4} + \frac{1}{12n}$ \citep{garg2020improved}.

Despite numerous studies devoted to their existence and computation, there is a conceptual and practical problem with the multiplicative approximations of MMS: they are very sensitive to agents' precise cardinal valuations. 
To illustrate, suppose $n=3$ and there are four goods $g_1,g_2,g_3,g_4$ that Alice values at $30, 39, 40, 41$ respectively. Her $1$-out-of-$3$ MMS is $40$, and thus a $\frac{3}{4}$ fraction guarantee can be satisfied by giving her the bundle $\{g_1\}$ or a bundle with a higher value.
But if her valuation of good $g_3$ changes slightly to $40+\varepsilon$ (for any $\varepsilon>0$), then  $\frac{3}{4}$ of her $1$-out-of-$3$ MMS is larger than $30$, the bundle $\{g_1\}$ is no longer acceptable for her.
Thus, the acceptability of a bundle (in this example $\{g_1\}$) might be affected by an arbitrarily small perturbation in the value of an \emph{irrelevant} good (i.e. $g_3$).


In the microeconomics literature, it is common to measure agents' preferences as \emph{ordinal rankings} of the bundles; even when utility functions are used, it is understood that they only represent rankings.
From this viewpoint, 
the set of acceptable bundles should only depend on the ranking of the bundles, and should not be affected by changes in valuations that---similar to the $\varepsilon$ change in the value of $g_3$---do not affect this ranking.
According to this principle, \citet{budish2011combinatorial} suggested the 1-out-of-$(n+1)$ MMS as a relaxation of the 1-out-of-$n$ MMS. 
In the above example, $1$-out-of-$4$ MMS fairness can be satisfied by giving Alice $\{g_1\}$ or a better bundle; small inaccuracies or noise in the valuations do not change the set of acceptable bundles.
Hence, this notion provides a more robust approach in evaluating fairness of allocations.

To date, it is not known if 1-out-of-$(n+1)$ MMS allocations are guaranteed to exist. 
We aim to find allocations of goods that guarantee $1$-out-of-$d$ MMS for some integer $d>n$. A $1$-out-of-$d$ MMS allocation guarantees to each agent a bundle that is at least as good as the worst bundle in the best $d$-partition. 

The aforementioned guarantee can be naturally generalized to  $\ell$-out-of-$d$ MMS \citep{babaioff2021competitive}, that guarantees to each agent the value obtained by partitioning the goods into $d$ bundles and selecting the $\ell$ least-valuable ones.
%
Therefore, we further investigate the $\ell$-out-of-$d$ MMS generalization that allows us to improve the fairness thresholds. 
%
The notion of $\ell$-out-of-$d$ MMS fairness is \emph{robust} in the sense that, a fair allocation remains fair even when each agent's utility function goes through an arbitrary monotonically-increasing transformation.
Given these notions, we ask the following questions: 
%
\begin{quote}
\emph{In the allocation of indivisible goods,
(a)
For what combinations of integers $\ell$ and $d$, can $\ell$-out-of-$d$ MMS allocations be guaranteed? and 
(b) For what integers $\ell$ and $d$ can  $\ell$-out-of-$d$ MMS allocations be computed in polynomial time?
}
\end{quote}


\subsection{Our Contributions}
We investigate the existence and computation of ordinal MMS approximations and make the several contributions.

In \textbf{Section \ref{sec:goods-lone}},
 we prove the existence of $\ell$-out-of-$d$ MMS allocation of goods when $d\geq \goodsapp$ (Theorem \ref{thm:l-out-of-d-existence}).
In particular, $1$-out-of-$\floor{3n/2}$ MMS, $2$-out-of-$\floor{5n/2}$ MMS, 
$3$-out-of-$\floor{7n/2}$ MMS, and so on, are all guaranteed to exist.
This finding  generalizes the previously known existence result of $1$-out-of-$\ceil{3n/2}$ MMS \citep{hosseini2021mms}. 

The proof uses an algorithm which, given lower bounds on the $\ell$-out-of-$d$ MMS values of the agents, returns an $\ell$-out-of-$d$ MMS allocation. The algorithm runs in polynomial time given the agents' lower bounds. However, computing the exact  $\ell$-out-of-$d$ MMS values is NP-hard. In the following sections we propose two solutions to this issue.

In \textbf{Section \ref{sec:goods-poly}}, we 
present polynomial-time algorithms 
that find an $\ell$-out-of-$(d + o(n))$ MMS-fair allocation, where $d=(\ell+\frac{1}{2})n$.
Specifically, for $\ell=1$, we present a polynomial-time algorithm 
for finding a 1-out-of-$\ceil{3n/2}$ MMS allocation (Theorem~\ref{thm:3n-2poly});
this matches the existence result for 1-out-of-$\floor{3n/2}$ MMS up to an additive gap of at most 1.
For $\ell>1$, we present a different polynomial-time algorithm 
for finding a 1-out-of-$\ceil{(\ell+\frac{1}{2})n + O(n^{2/3})}$ MMS allocation (Theorem~\ref{thm:goods-approx}).

In \textbf{Appendix \ref{sec:simulation}},
we conduct simulations with valuations generated randomly from various distributions.
For several values of $\ell$, we compute a lower bound on the $\ell$-out-of-$\goodsapp$ MMS guarantee using a simple greedy algorithm.
We compare this lower bound to an upper bound on the $(\frac{3}{4}+\frac{1}{12n})$-fraction MMS guarantee, which is currently the best known worst-case multiplicative MMS approximation.%
\footnote{
In general, ordinal and multiplicative approximations are incomparable from the theoretical standpoint---each of them may be larger than the other in some instances (see Appendix \ref{sec:simulation}). Therefore, we compare them through simulations using synthetic data. 
}
We find that, for any $\ell\geq 2$, when the number of goods is at least $\approx 20 n$, the lower bound on the ordinal approximation is better than the upper bound on the multiplicative approximation. 
This implies that, in practice, the algorithm of  Section \ref{sec:goods-lone} can be used with these lower bounds to attain an allocation in which each agent receives a value that is significantly better than the theoretical guarantees.

\subsection{Techniques}
At first glance, it would seem that the techniques used to attain $2/3$ approximation of MMS should also work for achieving $1$-out-of-$\floor{3n/2}$ MMS allocations, since both guarantees approximate the same value, namely, the $\frac{2}{3}$ approximation of 
the ``proportional share'' ($\frac{1}{n}$ of the total value of all goods).
In \textbf{Appendix \ref{sec:negative}} we present an example showing that this is not the case, and thus, achieving ordinal MMS approximations requires new  techniques. In this section, we briefly describe the techniques that we utilize to achieve ordinal approximations of MMS.

\paragraph{Lone Divider.}
To achieve the existence result for any $\ell\geq 1$,  we use a variant of the \emph{Lone Divider} algorithm, which was 
first presented by \citet{kuhn1967games} for finding a proportional allocation of a divisible good (also known as a ``cake''). Recently, it was shown that the same algorithm can be used for allocating indivisible goods too.
When applied directly, the Lone Divider algorithm finds only an $\ell $-out-of-$((\ell +1)n-2)$ MMS allocation \citep{aigner2022envy},
which for small $\ell$ is substantially worse than our target approximation of 
$\ell $-out-of-$\goodsapp$.
We overcome this difficulty by adding constraints on the ways in which the `lone divider' is allowed to partition the goods, as well as arguing on which goods are selected to be included in each partition (see Section \ref{sec:goods-lone}).

\paragraph{Bin Covering.}
To develop a polynomial-time algorithm when $\ell = 1$, 
we extend an algorithm of \citet{csirik1999two}
for the \emph{bin covering} problem---a dual of the more famous
\emph{bin packing} problem \citep{johnson1973near}.
In this problem, the goal is to fill as many bins as possible with items of given sizes, where the total size in each bin must be above a given threshold.
This problem is NP-hard, but \citet{csirik1999two} presents a polynomial-time $2/3$ approximation. 
This algorithm cannot be immediately applied to the fair division problem since the valuations of goods are \textit{subjective}, meaning that agents may have different valuations of each good.
We adapt this technique to handle subjective valuations.

\section{Related Work}
\label{sec:related}

\subsection{Maximin Share}
The idea of using the highest utility an agent could obtain if all other agents had the same preferences as a benchmark for fairness, originated in the economics literature \cite{Moulin1990Uniform,Moulin1992Welfare}.
It was put to practice in the context of course allocation by \citet{budish2011combinatorial}, where he introduced the ordinal approximation to MMS, and showed a mechanism that guarantees $1$-out-of-$(n+1)$ MMS to all agents by adding a small number of excess goods.
In the more standard fair division setting, in which adding goods is impossible,
the first non-trivial ordinal approximation was 1-out-of-$(2n-2)$ MMS \citep{aigner2022envy}.
\citet{hosseini2021mms} studied the connection between guaranteeing 1-out-of-$n$ MMS for $2/3$ of the agents and the ordinal approximations for \textit{all} agents. The implication of their results is the existence of $1$-out-of-$\ceil{3n/2}$ MMS allocations and a polynomial-time algorithm for $n<6$.
%
Whether or not $1$-out-of-$(n+1)$ MMS can be guaranteed without adding excess goods remains an open problem to date.

The generalization of the maximin share to arbitrary $\ell\geq 1$ was first introduced by
\citet{babaioff2019fair,babaioff2021competitive}, and further studied by \citet{segal2020competitive}.  
They presented this generalization as a natural fairness criterion for agents with different entitlements.
The implication relations between $\ell$-out-of-$d$ MMS-fairness guarantees for different values of $\ell$ and $d$ were characterized by \citet{segal2019maximin}.
Recently, the maximin share and its ordinal approximations have also been applied to some variants of the \emph{cake-cutting} problem \citep{ElkindSeSu21,ElkindSeSu21b,ElkindSeSu21c,bogomolnaia2022guarantees}.

\subsection{Multiplicative MMS Approximations}
The multiplicative approximation to MMS originated in the computer science literature \citep{procaccia2014fair}.
The non-existence of MMS allocations \citep{kurokawa2018fair} and its intractability \citep{Bouveret2016,woeginger1997polynomial} have given rise to a number of approximation techniques. 

These algorithms guarantee that each agent receives an approximation of their maximin share threshold. 
The currently known algorithms guarantee $\beta \geq 2/3$ \citep{kurokawa2018fair,amanatidis2017approximation,garg2018approximating} and $\beta \geq 3/4$ \citep{ghodsi2018fair,garg2020improved} in general, and $\beta \geq 7/8$ \citep{amanatidis2017approximation} as well as $\beta\geq 8/9$ \cite{gourves2019maximin} when there are only three agents.
There are also MMS approximation algorithms for settings with constraints, such as when the goods are allocated on a cycle and each agent must get a connected bundle \citep{truszczynski2020maximin}.
\citet{mcglaughlin2020improving} showed an algorithm for approximating the maximum Nash welfare (the product of agents' utilities), which also attains a fraction $1/(2n)$ of the MMS.

Recently, \citet{nguyen2017approximate} gave a Polynomial Time Approximation Scheme (PTAS) for a notion defined as \textit{optimal-MMS}, that is, the largest value, $\beta$, for which each agent receives at least a fraction $\beta$ of its MMS. Since the number of possible partitions is finite, an optimal-MMS allocation always exists, and it is an MMS allocation if $\beta \geq 1$. However, an optimal-MMS allocation may provide an arbitrarily bad ordinal MMS guarantee. \citet{Searns_Hosseini_2020,hosseini2021mms} show that for every $n$, there is an instance with $n$ agents in which under \textit{any} optimal-MMS allocation only a constant number of agents ($\leq 4$)  receive their MMS value.

\subsection{Fairness Vased on Ordinal Information}
An advantage of the ordinal MMS approximation is that it depends only on the ranking over the bundles. Other fair allocation algorithms with this robustness property are the Decreasing Demands algorithm of \citet{herreiner2002simple}, the Envy Graph algorithm of \citet{lipton2004approximately}, and the UnderCut algorithm of \citet{Brams2012Undercut}.

\citet{amanatidis2016truthful,halpern2021fair} study an even stronger robustness notion, where the agents report only a ranking over the \emph{goods}. Their results imply that, in this setting, the highest attainable multiplicative approximation of MMS is $\Theta(1/\log n)$.

\citet{menon2020algorithmic} define a fair allocation algorithm as \emph{stable} if it gives an agent the same value even if the agent slightly changes his cardinal valuations of goods, as long as the ordinal ranking of goods remains the same. They show that most existing algorithms are not stable, and present an approximately-stable algorithm for the case of two agents.

Finally, robustness has been studied also in the context of \emph{fair cake-cutting}. \citet{aziz2014cake} define an allocation \emph{robust-fair} if it remains fair even when the valuation of an agent changes, as long as its ordinal information remains unchanged.
\citet{edmonds2011cake} study cake-cutting settings in which agents can only cut the cake with a finite precision.

\section{Preliminaries} \label{sec:prel}

\subsection{Agents and Goods}
\label{sub:agents}
Let $N = [n] := \{1,\ldots, n\}$ be a set of agents and $M$ denote a set of $m$ indivisible goods. 
We denote the value of agent $i\in N$ for good $g\in M$ by $v_{i}(g)$. 
We assume that the valuation functions are \textit{additive}, that is, for each subset $G\subseteq M$, $v_{i}(G) = \sum_{g\in G} v_{i}(g)$, and  $v_i(\emptyset)=0$.%
\footnote{In Appendix~\ref{app:responsive} we complement our results with a non-existence result for the more general class of \textit{responsive} preferences.}

An \emph{instance} of the problem is denoted by $I = \ins{N, M, V}$, where $V = (v_1, \ldots, v_n)$ is the valuation profile of agents. We assume all agents have a non-negative valuation for each good $g\in M$, that is, $v_i(g) \geq 0$. 
An \emph{allocation} $A = (A_1, \ldots, A_n)$ is an $n$-partition of $M$ that allocates the bundle of goods in $A_{i}$ to each agent $i\in N$.

It is convenient to assume that the number of goods is sufficiently large. Particularly, some algorithms implicitly assume that $m\geq n$, while some algorithms implicitly assume that $m\geq \ell \cdot n$.
These assumptions are without loss of generality, since if $m$ in the original instance is smaller, we can just add dummy goods with a value of $0$ to all agents. 

\subsection{The Maximin Share}

For every agent $i\in N$ and integers $1\leq \ell < d$,
the \emph{$\ell$-out-of-$d$ maximin share of $i$ from $M$}, denoted $\mms{i}{\ell}{d}{M}$, is defined as
\begin{align*}
\mms{i}{\ell}{d}{M} := 
~~
\max_{\mathbf{P}\in \partition{M}{d}}
~~
\min_{Z\in \union{\mathbf{P}}{\ell}}
~~
v_i(Z)
\end{align*}
where the maximum is over all partitions of $M$ into $d$ subsets, and the minimum is over all unions of $\ell$ subsets from the partition. 
We say that an allocation $A$ is \emph{an $\ell$-out-of-$d$-MMS allocation} if  for all agents $i\in N$, $v_{i}(A_i) \geq \mms{i}{\ell}{d}{M}$.


Obviously $\mms{i}{\ell}{d}{M} \leq \frac{\ell}{d}v_i(M)$, and the equality holds if and only if $M$ can be partitioned into $d$ subsets with the same value. 
%
%
%
Note that $\mms{i}{\ell}{d}{M}$ is a weakly-increasing function of $\ell$ and  a weakly-decreasing function of $d$.

The value $\mms{i}{\ell}{d}{M}$ is at least as large, and sometimes larger than, $\ell\cdot \mms{i}{1}{d}{M}$. For example, suppose $\ell=2$, there are $d-1$ goods with value $1$ and one good with value $\varepsilon < 1$. Then $\mms{i}{2}{d}{M} = 1 + \varepsilon$ but $2\cdot \mms{i}{1}{d}{M} = 2\varepsilon$.

The maximin-share notion is scale-invariant in the following sense: if the values of each good for an agent, say $i$, are multiplied by a constant $c$, then agent $i$'s MMS value is also multiplied by the same $c$, so the set of bundles that are worth for $i$ at least $\mms{i}{\ell}{d}{M}$ does not change.

\subsection{The Lone Divider Algorithm}
\label{sub:lone-divider}
A general formulation of the Lone Divider algorithm, based on \citet{aigner2022envy}, is shown in Algorithm \ref{alg:lone-divider-general}. It accepts as input a set $M$ of items and a threshold value $t_i$ for each agent $i$. These values should satisfy the following condition for each agent $i\in N$.
\begin{definition}[Reasonable threshold]
\label{def:reasonable}
Given a set $M$, a value function $v_i$ on $M$, and an integer $n\geq 2$, \emph{a reasonable threshold for $v_i$} is a real number $t_i\in \mathbb{R}$ satisfying the following condition: For every integer $k\in\{0,\ldots,n-1\}$ and any $k$ disjoint subsets $B_1,\ldots,B_k \subseteq M$, if
\begin{align*}
\forall c\in[k]: v_i(B_c) < t_i,
\end{align*}
then there exists a partition of 
$M\setminus  \cup_{c\in[k]} B_c$
into 
$M_1 \cup  \cdots \cup  M_{n-k}$, such that
\begin{align*}
\forall j\in[n-k]: v_i(M_j)\geq t_i.
\end{align*}
Informally, if any $k$ unacceptable subsets are given away, then $i$ can partition the remainder into  $n-k$ acceptable subsets.
In particular, the case $k=0$ implies that agent $i$ can partition the original set $M$ into $n$ acceptable subsets.

Given an instance $I = \ins{N, M, V}$ with $N=[n]$, 
a vector $(t_i)_{i=1}^n$ of real numbers is called \emph{a reasonable threshold vector for $I$}
if $t_i$ is a reasonable threshold for $v_i$ for all $i\in N$.%
\footnote{
While we use the Lone Divider algorithm for allocating indivisible goods, it is a more general scheme that can also be used to divide chores or mixed items, divisible or indivisible. See \citet{aigner2022envy} for details.
}
\end{definition}

\begin{example} [\textbf{Reasonable threshold}]
Suppose $M$ is perfectly divisible (e.g. a cake), and let $t_i := v_i(M)/n$. This threshold is reasonable, since if some $k$ bundles with value less than $t_i$ are given away, the value of the remaining cake is more than $(n-k)t_i$. Since the cake is divisible, 
it can be partitioned into  $n-k$ acceptable subsets.
This does not necessarily hold when $M$ is a set of indivisible items; hence, finding reasonable thresholds for indivisible items setting is more challenging. $\blacksquare$
\end{example}

\begin{algorithm}[t]
\caption{
\label{alg:lone-divider-general}
The Lone Divider algorithm.
Based on \citet{kuhn1967games}.
}
\begin{algorithmic}[1]
\REQUIRE ~
  An instance $\ins{N, M, V}$ where $N$ is the set of agents, $M$ is the set of items, $V$ is the vector of agents' valuations; and
a reasonable threshold vector $(t_i)_{i=1}^n$ as denoted in Definition \ref{def:reasonable}.
\ENSURE A partition $M = A_1\cup \cdots\cup A_n$ such that $v_i(A_i)\geq t_i$ for all $i\in [n]$.


\STATE  
\label{step:ldg-cut}
Some arbitrary agent $a\in N$ is asked to
partition $M$ into $|N|$ disjoint subsets, $(Y_j)_{j\in N}$, with 
$\forall j\in N: v_a(Y_j)\geq t_a$.
\STATE 
\label{step:ldg-graph}
Define a bipartite graph $G$ with the agents of $N$ on the one side and the set $Y := \{Y_1,\ldots,Y_{|N|}\}$ on the other side. Add an edge  $(i, Y_j)$ whenever $v_i(Y_j)\geq t_i$. 
\STATE 
\label{step:ldg-efm}
Find a maximum-cardinality envy-free matching$^{\ref{ftn:efm}}$ in $G$.
Give each matched element in $Y$ to the agent paired to it in $N$.
\STATE 
\label{step:ldg-recurse}
Let $N \leftarrow $ the unallocated agents
and $M \leftarrow $ the unallocated objects. 
If $N\neq\emptyset$
go back to Step \ref{step:ldg-cut}.
\end{algorithmic}
\end{algorithm}

\paragraph{Algorithm Description} Algorithm \ref{alg:lone-divider-general} proceeds in the following way: in each step, a single remaining agent is asked to partition the remaining goods into acceptable bundles---bundles whose values are above the divider's threshold.
Then, all agents point at those bundles that are acceptable for them, and the algorithm finds an \textit{envy-free matching} in the resulting bipartite graph.%
\footnote{
\label{ftn:efm}
An \emph{envy-free matching} in a bipartite graph $(N\cup Y, E)$ is a matching in which each unmatched agent in $N$ is not 
adjacent to any matched element in $Y$.
The bipartite graph generated by the Lone Divider algorithm always admits a nonempty envy-free matching, and a maximum-cardinality envy-free matching can be found in polynomial time \citep{aigner2022envy}.
}
The matched bundles are allocated to the matched agents, and the algorithm repeats with the remaining agents and goods. 
It is easy to see that, if all threshold values $t_i$ are reasonable, then Lone Divider  guarantees agent $i$ a bundle with a value of at least $t_i$.
For example, 
when $M$ is a cake,
$t_i = v_i(M)/n$ is a reasonable threshold for every $i$, so Lone Divider can be used to attain a \emph{proportional cake-cutting} \citep{kuhn1967games}.

When $M$ is a set of indivisible goods, $t_i = \mms{i}{\ell}{[(\ell+1)n-2]}{M}$ is a reasonable threshold for every $\ell\geq 1$ \citep{aigner2022envy}, so these ordinal approximations can all be computed directly through the Lone Divider algorithm.
However, directly applying the Lone Divider algorithm cannot guarantee a better ordinal approximation, as we show next.

\begin{example}[\textbf{Execution of Algorithm \ref{alg:lone-divider-general}}]
\label{exm:lone-divider}
For simplicity, we present an example for $\ell=1$.
We show that, 
while Algorithm \ref{alg:lone-divider-general} can guarantee $1$-out-of-$(2n-2)$ MMS,
 it cannot guarantee $1$-out-of-$(2n-3)$ MMS. 
Suppose that there are $4n-6$ goods, and that all agents except the first divider value some $2n-3$ goods at $1-\varepsilon$ and the other $2n-3$ goods at $\varepsilon$
(see Figure \ref{fig:lone-divider}). Then the 1-out-of-$(2n-3)$ MMS of all these agents is $1$.

However, it is possible that the first divider takes an unacceptable bundle containing all $2n-3$ goods of value $\varepsilon$.
Then, no remaining agent can partition the remaining goods into $n-1$ bundles of value at least $1$. In this instance, it is clear that while $\mms{i}{\ell}{(2n-2)}{M}$ is a reasonable threshold,
$\mms{i}{\ell}{(2n-3)}{M}$  is not. $\blacksquare$

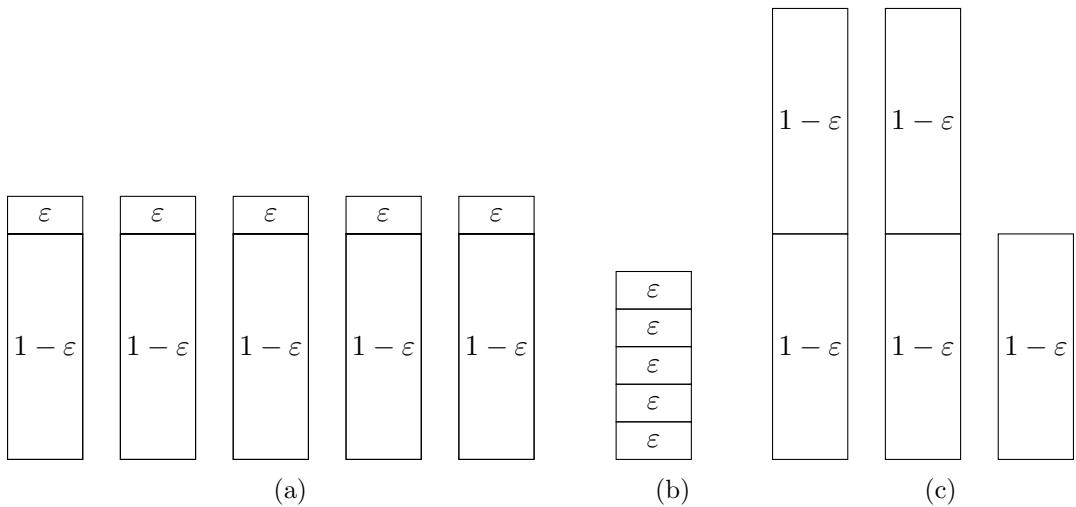
\begin{figure}
\begin{subfigure}[b]{0.5\textwidth}
\begin{tikzpicture}[scale=0.5]
\mmsbundle{0}{6}{$1-\varepsilon$}{7}{$\varepsilon$}
\mmsbundle{3}{6}{$1-\varepsilon$}{7}{$\varepsilon$}
\mmsbundle{6}{6}{$1-\varepsilon$}{7}{$\varepsilon$}
\mmsbundle{9}{6}{$1-\varepsilon$}{7}{$\varepsilon$}
\mmsbundle{12}{6}{$1-\varepsilon$}{7}{$\varepsilon$}
\end{tikzpicture}
\caption{}
\end{subfigure}
~~
\begin{subfigure}[b]{0.1\textwidth}
\begin{tikzpicture}[scale=0.5]
\itemrect{0}{0}{2}{1}{$\varepsilon$}
\itemrect{0}{1}{2}{1}{$\varepsilon$}
\itemrect{0}{2}{2}{1}{$\varepsilon$}
\itemrect{0}{3}{2}{1}{$\varepsilon$}
\itemrect{0}{4}{2}{1}{$\varepsilon$}
\end{tikzpicture}
\caption{}
\end{subfigure}
~~
\begin{subfigure}[b]{0.3\textwidth}
\begin{tikzpicture}[scale=0.5]
\itemrect{0}{0}{2}{6}{$1-\varepsilon$}
\itemrect{0}{6}{2}{6}{$1-\varepsilon$}
\itemrect{3}{0}{2}{6}{$1-\varepsilon$}
\itemrect{3}{6}{2}{6}{$1-\varepsilon$}
\itemrect{6}{0}{2}{6}{$1-\varepsilon$}
\end{tikzpicture}
\caption{}
\end{subfigure}
\caption{
\label{fig:lone-divider}
An illustration of the goods' values in Example \ref{exm:lone-divider}, for $n=4$.
\\
(a) The $2n-3=5$ MMS bundles of some agent.
\\
(b) The unacceptable bundle taken by the first divider.
\\
(c) The remaining goods, which cannot be combined into $n-1=3$ acceptable bundles.
}
\end{figure}

\end{example}

\section{Ordinal Approximation of MMS for Goods}
\label{sec:goods-lone}
In this section we prove the following theorem.

\begin{theorem} \label{thm:l-out-of-d-existence}
Given an additive goods instance, an $\ell$-out-of-$d$ MMS allocation always exists when 
$d = \goodsapp$.
\end{theorem}
The proof is constructive: 
we present an algorithm (Algorithm \ref{alg:existenceMMS}) for achieving the above MMS bound. 
Since the algorithm needs to know the exact MMS thresholds for each agent (which is NP-hard to compute), its run-time is not polynomial.
In Section \ref{sec:goods-poly} we present a different algorithm to compute $\ell$-out-of-$d$ MMS allocation when $\ell = 1$ in polynomial-time.

%

Algorithm \ref{alg:existenceMMS} starts with two normalization steps, some of which appeared in previous works and some are specific to our algorithm. 
For completeness, we describe the normalization steps in Sections \ref{sub:normalize-goods} and \ref{sub:ordering}.
The algorithm applies to the normalized instance an adaptation of the Lone Divider algorithm, in which the divider in each step must construct a \emph{balanced} partition. We explain this notion in Section \ref{sub:restricted-ld}.

\subsection{Scaling}
\label{sub:normalize-goods}

We start by scaling the valuations such that 
$\mms{i}{\ell}{d}{M} = \ell$ for each agent $i$. 
The scale-invariance property implies that such rescalings do not modify the set of bundles that are acceptable for $i$.
Then, for each $i$ we perform an additional scaling as follows.
\begin{itemize}
\item Consider a particular $d$-partition  attaining the maximum in the $\mms{i}{\ell}{d}{M}$ definition.
Call the $d$ bundles in this partition the \emph{MMS bundles} of agent $i$.
\item Denote the total value of the $\ell-1$ least-valuable MMS bundles by $x_i$ (or just $x$, when $i$ is clear from the context).
By definition, the value of the $\ell$-th MMS bundle must be exactly $\ell-x$, while the value of each of the other $d-\ell$ MMS bundles is at least $\ell-x$. 
\item For each MMS bundle with value larger than $\ell-x$, arbitrarily pick one or more goods and decrease their value until the value of the MMS bundle becomes exactly $\ell-x$. Note that this does not change the MMS value.
\end{itemize}

After the normalization, the sum of values of all goods is
\begin{align*}
v_i(M) = 
&
\ (d-\ell+1)\cdot (\ell-x_i) + x_i 
\\
=& \ \ell + (d-\ell)(\ell-x_i) 
= d + (d-\ell)(\ell-1-x_i).
\end{align*}

Since $d = \goodsapp \geq (\ell+\frac{1}{2})n -\frac{1}{2} = \ell n + n/2 -1/2 $,
\begin{align}
\notag
v_i(M)
\geq &
\ (\ell n + n/2 -1/2 ) 
 + (\ell n + n/2 -1/2 -\ell)(\ell-1-x_i)
\\
\label{eq:total-value}
=&
\ n\cdot \ell
 ~~+~~ (n-1)\cdot \ell(\ell-1-x_i)
 ~~+~~ (n-1)\cdot (\ell-x_i)/2.
\end{align}

The goal of the algorithm is to give each agent $i$ a bundle $A_i$ with $v_i(A_i)\geq \ell$. We say that such a bundle is \emph{acceptable} for $i$.

\begin{example}[\textbf{Scaling}]
\label{exm:x}
To illustrate the parameter $x$, consider the following two instances with $n=5, \ell=3$ and  $d=\goodsapp=17$.
\begin{enumerate}
\item There are $17$ goods with the value of $1$. 
\item There are $16$ goods valued $1.2$ and one good with the value of $0.6$.
\end{enumerate}
Here, each MMS bundle contains a single good. In both cases, the value of every $3$ goods is at least $3$. 
In the first case $x=2$
and the total value is $5\cdot 3 + 4\cdot 3\cdot 0 + 4\cdot 1/2 = 17$.
In the second case, $x=1.8$
and the total value is $5\cdot 3 + 4\cdot 3\cdot 0.2 + 4\cdot 1.2/2 = 19.8$. $\blacksquare$
\end{example}

\subsection{Ordering the Instance}
\label{sub:ordering}
As in previous works \citep{Bouveret2016,barman2017approximation,garg2018approximating,huang2021algorithmic}, 
we apply a preliminary step in which the instance is \emph{ordered}, i.e., 
$v_i(g_1)\geq \cdots \geq v_i(g_m)$ for each agent $i\in N$.
Ordering is done as follows:
\begin{itemize}
\item Index the goods in $M$ arbitrarily $g_1,\ldots, g_m$.
\item Tell each agent $i$ to adopt, for the duration of the algorithm, a modified value function that assigns, to each good $g_j$, the value of the $j$-th most valuable good according to $i$.
For example, the new $v_i(g_1)$ should be the value of $i$'s most-valuable good; the new $v_i(g_m)$ should be the value of $i$'s least-valuable good; etc.
Ties are broken arbitrarily.
\end{itemize} 

During the execution of the algorithm, each agent answers all queries according to this new value function.
For example, an agent asked whether the bundle $\{g_1,g_4,g_5\}$ is acceptable, should answer whether the bundle containing his best good, 4th-best good and 5th-best good is acceptable.
Once the algorithm outputs an allocation, it can be treated as a \emph{picking sequence} in which, for example, an agent who receives the bundle $\{g_1,g_4,g_5\}$ has the first, fourth and fifth turns. It is easy to see that such an agent receives a bundle that is at least as good as the bundle containing her best, 4th-best and 5th-best goods. Hence, if the former is acceptable then the  latter is acceptable too. 


Clearly, given an \textit{un}ordered instance, 
its corresponding ordered instance 
can be generated in polynomial time
(for each agent $i\in [n]$, we need $O(m\log m)$ steps for ordering the valuations).
Given an allocation for the ordered instance, one can compute the allocation for the corresponding unordered instance in time $O(n)$,
using the picking-sequence  described above. 
%

\subsection{Restricted Lone Divider}
\label{sub:restricted-ld}

\begin{algorithm}[t]
\caption{
\label{alg:(L+1/2)n}
Finding an
$\ell$-out-of-$\goodsapp$ 
MMS allocation.
}
\begin{algorithmic}[1]
\REQUIRE An instance $\ins{N, M, V}$ and an integer $\ell \geq 1$.
\ENSURE An 
$\ell$-out-of-$\goodsapp$ 
MMS allocation.
\STATE  Scale the valuations of all agents as explained in Section \ref{sub:normalize-goods}.
\STATE  Order the instance as explained in Section \ref{sub:ordering}.
\STATE 
Run the Lone Divider algorithm (Algorithm \ref{alg:lone-divider-general}) with 
threshold values $t_i = \ell$ for all $i\in N$,
with the restriction that, in each partition made by the lone divider, all bundles must be $\ell$-balanced (Definition \ref{def:balanced-bundle}).
\end{algorithmic}
\label{alg:existenceMMS}
\end{algorithm}
In Section \ref{sub:lone-divider} we illustrated the limitations of the plain Lone Divider algorithm 
(Algorithm \ref{alg:lone-divider-general}).
We can improve its performance by restricting the partitions that the lone divider is allowed to make in Step 1 of Algorithm \ref{alg:lone-divider-general}.
Without loss of generality, we may assume (by adding dummy goods if needed) that $m\geq n\cdot \ell$.

For every $l\range{1}{\ell}$, denote
$G_l^n := \{g_{(l-1)n+1},\ldots,g_{l n}\}$. In other words, $G_1^n$ contains the $n$ most-valuable goods; $G_2^n$ contains the $n$ next most-valuable goods; and so on.
Since the instance is ordered, these sets are the same for all agents.
\begin{definition}[$\ell$-balanced bundle]
\label{def:balanced-bundle}
Given an ordered instance and an integer $\ell\geq 1$, a nonempty bundle $B\subseteq M$ is called \emph{$\ell$-balanced} if
\begin{itemize}
\item $B$ contains exactly one good from $G^n_1$.
\item If $|B|\geq 2$, then $B$ contains exactly one good from $G^n_2$.
\item If $|B|\geq 3$, then $B$ contains exactly one good from $G^n_3$.
\item $\ldots$
\item If $|B|\geq \ell$, then $B$ contains exactly one good from $G^n_{\ell}$.
\end{itemize}

Note that an $\ell$-balanced bundle contains at least $\ell$ goods.
The definition of $\ell$-balanced bundles only constrains the allocation of the first $\ell n$ goods; there may be arbitrarily many additional goods in $M \setminus \bigcup_{i = 1}^{\ell} G_{i}^{n}$, and they may be allocated arbitrarily.

\end{definition}
Algorithm \ref{alg:(L+1/2)n} requires the lone divider to construct a partition in which all $n$ bundles are $\ell$-balanced.

\begin{example}[\textbf{$\ell$-balanced bundles}]
Suppose there are five agents ($n=5$) and $m=20$ goods, where the value of each good $j\in[20]$ is precisely $j$ for all agents. Then,
a $1$-balanced bundle must contain a good $j \in \{20,19, 18, 17,16\}$;
a $2$-balanced bundle must contain a good from $\{20,19, 18, 17,16\}$ and a good from $\{15, 14, 13, 12, 11\}$;
a $3$-balanced bundle must contain, in addition to these, a good from $\{10, 9, 8, 7 ,6\}$; and so on. $\blacksquare$
\end{example}

\subsection{Construction for a Single Divider}
In order to prove the correctness of Algorithm \ref{alg:(L+1/2)n},
it is sufficient to prove that 
the threshold value $t_i = \ell$ is a \textit{reasonable threshold} (see Definition \ref{def:reasonable}) for each agent $i$, 
with the additional restriction that all bundles should be $\ell$-balanced.

 
To do this, it is sufficient to consider a single divider, Alice. We denote her normalized ordered value measure by $v$, and the sum of her $\ell-1$ least-valuable MMS bundles by $x$.
We consider a particular MMS partition for Alice, and refer to the bundles in this partition as the \emph{MMS bundles}.

Assume that $k$ unacceptable bundles $(B_c)_{c=1}^k$ have already been given to other agents
and that all these bundles are $\ell$-balanced.
Therefore, for each $c\in[k]$, it must be that $v(B_c)<\ell$.\footnote{
Recall that the Lone Divider algorithm allocates bundles using an envy-free matching. This means that all bundles allocated before Alice's turn are unacceptable to Alice.
} 
We have to prove that Alice can use the remaining goods to construct $n-k$ acceptable bundles that are also $\ell$-balanced.
Particularly, we prove below that Alice can construct $n-k$ acceptable bundles, each of which contains exactly $1$ remaining good from each of $G^n_1\ldots,G^n_{\ell}$.

\subsection{Main Idea: Bounding the Waste}
Given a bundle $B_a$, denote its \emph{waste} by $w(B_a) := v(B_a) - \ell$. 
This is the value the bundle contains beyond the acceptability threshold of $\ell$.
Note that the waste of acceptable bundles is positive and that of unacceptable bundles is negative.
The total initial value for Alice is given by 
\eqref{eq:total-value}. The total waste she can afford in her partition is therefore
\begin{align*}
v(M) - n\cdot \ell
= 
(n-1)\cdot (\ell-x)/2
+
(n -1)\cdot \ell(\ell-1-x).
 \end{align*}
The first term implies that she can afford an average waste of $(\ell-x)/2$ for 
$n-1$
bundles; the second term implies that she can afford an average waste of $\ell(\ell-1-x)$ for $n-1$ bundles. 

\begin{example}[\textbf{Bounding the waste}]
Consider Example \ref{exm:x}.
In case (1), the total value is $17$ and we need $5$ bundles with a value of $3$, so the affordable waste is $2$. The average over $4$ bundles is $0.5 = (3-2)/2 + 3\cdot 0$.
In case (2), the total value is $19.8$, so the affordable waste is $4.8$. The average over $4$ bundles is $1.2 = (3-1.8)/2 + 3\cdot (0.2)$.
In both cases, if there are $4$ acceptable bundles with that amount of waste, then the remaining value is exactly $3$, which is sufficient for an additional acceptable bundle. $\blacksquare$
\end{example}

The following lemma formalizes this observation.
\begin{lemma}
\label{lem:waste}
Suppose there exists a partition of $M$ into 
\begin{itemize}
\item 
Some $t\geq 0$ bundles with an average waste of at most $(\ell-x)/2 + \ell(\ell-1-x)$;
\item A subset $S$ of remaining goods, with $v(S)<\ell$.
\end{itemize}
Then $t \geq n$.
\end{lemma}

\begin{proof}
For brevity, we denote $w := (\ell-x)/2 ~+~  \ell(\ell-1-x)$.
The \emph{total} value of the bundles equals their number times their \emph{average} value. 
So the total value of the $t$ bundles is at most $t\cdot \ell + t\cdot w$.
After adding $v(S) < \ell$  for the remaining goods, the sum equals $v(M)$, so
\begin{alignat*}{4}
& (t+1)\cdot \ell + t\cdot w > v(M) 
\\
\geq & n\cdot \ell + (n-1)\cdot w &&& \text{~~~~(by \eqref{eq:total-value}).}
\end{alignat*}
Therefore, at least one of the two terms in the top expression must be larger than the corresponding term in the bottom expression. This means that either $(t+1)\ell > n\ell$, or $t w  > (n-1)w$. Both options imply $t\geq n$.
\end{proof}

\begin{remark}
\label{rem:waste}
The value of each MMS bundle is at most $\ell-x$. Therefore, any bundle that is the union of exactly $\ell$ such MMS bundles has a waste of at most $\ell(\ell-x)-\ell = \ell(\ell-1-x)$ and thus it satisfies the upper bound of Lemma \ref{lem:waste}. 
In particular, this is satisfied for every bundle with at most $\ell$ goods.
\end{remark}

Below we show how Alice can find a partition in which the average waste is upper bounded as in Lemma \ref{lem:waste}. This partition will consist of the following bundles:
\begin{itemize}
\item The $k$ previously-allocated bundles, with wsate $<0$ (since they are unacceptable);
\item Some newly-constructed bundles with exactly $\ell$ goods and waste
$\leq \ell(\ell-1-x)$ (by Remark \ref{rem:waste});
\item Some newly-constructed bundles with waste at most $(\ell-x)/2$;
\item Some pairs of bundles, where the waste in one is larger than $(\ell-x)/2$ but the waste in the other is smaller than $(\ell-x)/2$, such that the average is at most $(\ell-x)/2$.
\end{itemize}

\subsection{Step 0: Bundles with Exactly $\ell$ Goods}
Recall that before Alice's turn, some $k$ bundles have been allocated, with a value of less than $\ell$. Hence, their waste is less than $0$.
Since these bundles are $\ell$-balanced, 
they contain exactly one good from each of $G^n_{1},\ldots,G^n_{\ell}$
(and possibly some additional goods).

Therefore, exactly $n-k$ goods are available in each of $G^n_1\ldots,G^n_{\ell}$. 

Next, Alice checks all the $\ell$-tuples containing one good from each of $G^n_1\ldots,G^n_{\ell}$ (starting from the highest-valued goods in each set). If the value of such an $\ell$-tuple is at least $\ell$, then it is acceptable and its waste is at most $\ell(\ell-1-x)$ by Remark \ref{rem:waste}.

After Step 0, there are
some $k' \geq k$ bundles with a waste of at most $\ell(\ell-1-x)$, each of which contains exactly one good from each of $G^n_1\ldots,G^n_{\ell}$. 
Of these, $k$ are previously-allocated bundles, and $k'-k$ are newly-constructed acceptable bundles.
In each of $G^n_1\ldots,G^n_{\ell}$, there remain exactly $n-k'$ goods. 
The total value of each $\ell$-tuple of remaining goods from $G^n_1\ldots,G^n_{\ell}$ is less than $\ell$.
Alice will now construct from them some $n-k'$ bundles with an average waste of at most $(\ell-x)/2 + \ell(\ell-1-x)$.
Lemma \ref{lem:waste} implies the following lemma on the remaining goods (the goods not in these $k'$ bundles):
\begin{lemma}
\label{lem:waste2}
Suppose there exists a partition of the remaining goods into 
\begin{itemize}
\item 
Some $t\geq 0$ bundles with an average waste of at most $(\ell-x)/2 + \ell(\ell-1-x)$;
\item A subset $S$ of remaining goods, with $v(S)<\ell$.
\end{itemize}
Then $t \geq n-k'$.
\end{lemma}

Alice's strategy branches based on the number of \emph{high-value goods}.

\subsection{High-value Goods}
\label{sub:high-value-goods}
We define high-value goods as goods $g$ with $v(g)>(\ell-x)/2$. 
Denote by $h$ the number of high-valued goods in $M$.
Since the instance is ordered, goods $g_1,\ldots, g_h$ are high-valued.
All MMS bundles are worth at most $\ell-x$, and therefore may contain at most one high-value good each.
Since the number of MMS bundles is $\ell n + n/2$, we have $h\leq \ell n + n/2$.

For each $j\in[h]$, we denote
\begin{itemize}
\item $M_j$ := the MMS bundle containing $g_j$. Since the value of all MMS bundles is at most $(\ell-x)$, 
each MMS bundle contains at most one high-value good, so the $M_j$ are all distinct.
\item $R_j$ := the \emph{remainder set} of $g_j$, i.e., the set $M_j \setminus \{g_j\}$.
\item $r_j := v(R_j)$.
\end{itemize}

We consider three cases, based on the number of high-value goods.

\paragraph{\textbf{Case \#1}: $h\leq \ell n$.}
This means that all high-value goods are contained in $G^n_1\cup\cdots \cup G^n_{\ell}$,
so after removing  $(B_c)_{c=1}^{k'}$, at most $\ell n-\ell k'$ high-value goods remain---at most $n-k'$ in each of $G^n_1\ldots,G^n_{\ell}$. 
Alice constructs the required bundles by \emph{bag-filling}---a common technique in MMS approximations (e.g. \citet{garg2018approximating}).

\begin{itemize}
\item Repeat at most $n-k'$ times:
\begin{itemize}
\item Initialize a bag with 
a good from each of $G^n_1\ldots,G^n_{\ell}$
(Step 0 guarantees that the total value of these goods is less than $\ell$).
\item Fill the bag with goods from outside $G^n_1\ldots,G^n_{\ell}$. Stop when either no such goods remain, or the bag value raises above $\ell$.
\end{itemize}
\end{itemize}
Since all goods used for filling the bag have a value of at most $(\ell-x)/2$, the waste of each constructed bundle is at most $(\ell-x)/2$.
By construction, all these bundles are acceptable except the last one.
Apply Lemma \ref{lem:waste2} with $S$ being the set of goods remaining in the last bag, 
and $t$ being the number of acceptable bundles constructed by bag-filling.
The lemma implies that $t\geq n-k'$. 

\paragraph{\textbf{Case \#2}: $k' \geq n/2$.}
Alice uses bag-filling as in Case \#1.

Here, the waste per constructed bundle might be more than $(\ell-x)/2$. However, since the value of a single good is at most $\ell-x$, the waste of each constructed bundle is at most $\ell-x$.

In each of the $k'$ bundles of Step 0, the waste is at most $\ell(\ell-1-x)$. Since $k'\geq n/2 \geq n-k'$, the \emph{average} waste per bundle is at most $\ell(\ell-1-x)+(\ell-x)/2$.
Hence Lemma \ref{lem:waste2} applies, and at least $n-k'$ acceptable bundles are constructed.

\paragraph{\textbf{Case \#3}: $h > \ell n$ and $k'<n/2$.}
In this case, Alice will have to construct some bundles with waste larger than $(\ell-x)/2$. However, she will compensate for it by constructing a similar number of bundles with waste smaller than $(\ell-x)/2$, such that the average waste per bundle remains at most $(\ell-x)/2$.

After removing  $(B_c)_{c=1}^{k'}$, exactly $h-\ell k'$ high-value goods remain. They can be partitioned into two subsets:
\begin{itemize}
\item $H_+ := $ the $(n-k')\ell$ top remaining goods --- those contained in $G^n_1\cup\cdots \cup G^n_{\ell}$;  exactly $n-k'$ in each of $G^n_1\ldots,G^n_{\ell}$. By assumption, $n-k' > n/2$.
\item $H_- := $ the other $h-\ell n$ high-value goods --- those not contained in $G^n_1\cup\cdots \cup G^n_{\ell}$.
Since $h\leq \ell n + n/2$, the set $H_-$ contains at most $n/2$ goods. 
\end{itemize}

This is the hardest case; to handle this case, we proceed to Step 1 below.

\subsection{Step 1: Bundling High-value Goods.}
Alice constructs 
at most
$|H_-|$ bundles as follows.
\begin{itemize}
\item Repeat while $H_-$  is not empty:
\begin{itemize}\item Initialize a bag with the lowest-valued remaining good from each of $G^n_1\ldots,G^n_{\ell}$
(Step 0 guarantees that their total value is less than $\ell$).
\item Fill the bag with goods from $H_-$, until the bag value raises above $\ell$.
\end{itemize}
\end{itemize}
Note that $|H_-| \leq n/2 < n-k' = |G^n_1| = \ldots = |G^n_{\ell}|$, so as long as $H_-$ is nonempty, 
each of $G^n_1\ldots,G^n_{\ell}$ is nonempty too, and Alice can indeed repeat.
By construction, all filled bags except the last one are valued at least $\ell$; it remains to prove that the number of these bags is sufficiently large.

Let $s$ be the number of acceptable bundles  constructed once $H_-$ becomes empty.
Note that, in addition to these bundles, there may be an \emph{incomplete bundle} --- the last bundle, whose construction was terminated while its value was still below $\ell$.

Let $P_+ \subseteq H_+$ be the set of $s \ell$ goods from $H_+$ in the acceptable bundles.
Denote by $P_- \subseteq H_-$ the set of $s$ goods from $H_-$ that were added \emph{last} to these $s$ bundles (bringing their value from less-than-$\ell$ to at-least-$\ell$).
Note that the waste in each of these $s$ bundles might be larger than $(\ell-x)/2$, but it is at most the value of a single good from $P_-$, so the total waste is at most $\sum_{j  \in P_-} v(g_{j})$.

After this step, 
besides the $s$ acceptable bundles, there are 
some $(n-k'-s)\ell$ high-value goods remaining in $H_+$
(some $\ell$ of these goods are possibly in the incomplete bundle, if such a bundle exists).
Alice now has to construct from them some $n-k'-s$ acceptable bundles.

\subsection{Step 2: Using the Remainders.}
Alice now constructs acceptable bundles by bag-filling. She initializes each bag with 
the incomplete bundle from Step 1 (if any), or with an $\ell$-tuple of unused goods from  $G^n_1\ldots,G^n_{\ell}$.
Then, she fills the bag with low-value goods from the following \emph{remainder sets}:

\begin{itemize}
\item 
There are $\ell k$ remainder sets that correspond to the $\ell k$ goods allocated within the 
$k$ unacceptable bundles $(B_c)_{c=1}^k$.
We denote them by $R^U_{c,1},\ldots,R^U_{c,\ell}$ 
and their values by 
$r^U_{c,1},\ldots,r^U_{c,\ell}$ 
for $c\in[k]$.
\item  
There are $\ell s + s$ remainder sets that correspond
to the 
$\ell s$ high-value goods in $P_+$ and the $s$ high-value goods in $P_-$.
We denote them by $R^P_{j}$ 
and their values by $r^P_{j}$ 
for $j\in P_+ \cup P_-$.
\end{itemize}

By definition, the total value of all these remainders is:
\begin{align*}
\textsc{Total-Remainder-Value}
&=
v
\left(
\bigcup_{j \in P_+ \cup P_-} R^P_{j}
~~\cup~~ 
\bigcup_{c=1}^{k} \bigcup_{l=1}^{\ell} R^U_{c,l}
\right) 
\\
&= 
\sum_{j \in P_+ \cup P_-} r^P_{j} 
~~+~~ 
\sum_{c=1}^{k} \sum_{l=1}^{\ell}
 r^U_{c,l}.
\end{align*}
For each remainder-set $R_j$, denote by $R'_j$, the subset of $R_j$ that remains after removing the at most $k$ unacceptable bundles $(B_c)_{c=1}^k$ with more than $\ell$ goods.%
\footnote{
Bundles $B_c$ with at most $\ell$ goods do not consume anything from the remainder-sets $R_j$, since they contain only high-value goods from $G^n_1,\ldots,G^n_{\ell}$.
From the same reason, the $k'-k$ acceptable bundles constructed in Step 0 do not consume anything from the remainder sets.
}
Each unacceptable bundle $B_c$ contains, in addition to the $\ell$ high-value goods $g_{c,l}$ for $l\range{1}{\ell}$, some low-value goods with a total value of less than $\sum_{l=1}^\ell r^U_{c,l}$ (since the total value of the unacceptable bundle is less than $\ell$).
Therefore, the total value of low-value goods included in these unacceptable bundles is at most $\sum_{c=1}^k \sum_{l=1}^\ell r^U_{c,l}$ (equality holding iff $k=0$). Therefore, the total remaining value satisfies
\begin{align}
\notag
\textsc{Total-Remainder-Value}
&=
v
\left(
\bigcup_{j  \in P_+ \cup P_-} R'^P_{j}
~~\cup~~ 
\bigcup_{c=1}^{k} \bigcup_{l=1}^{\ell} R'^U_{c,l}
\right) 
\\
\notag
&\geq
\left(
\sum_{j  \in P_+ \cup P_-} r^P_{j} 
~~+~~ 
\sum_{c=1}^{k} \sum_{l=1}^{\ell} r^U_{c,l}
\right)
-
\sum_{c=1}^{k} \sum_{l=1}^{\ell} r^U_{c,l}
\\
&=
\sum_{j  \in P_+ \cup P_-} r^P_{j}
.
\label{eq:total-value-remainders}
\end{align}

The bag-filling proceeds as follows. 

\begin{enumerate}
\item Initialize $a := 1$.
\item Initialize a bag with either the incomplete bundle from Step 1 (if any), 
or some $\ell$ unused top goods.
We denote the $\ell$ top goods used for initializing bag $a$
by $g_{j[a,1]}\in G^n_1,\ldots,g_{j[a,\ell]}\in G^n_{\ell}$.
\item Add to the bag the remainder-sets $R'^P_{j}$ and $R'^U_{c,l}$ in an arbitrary order. Stop when either no such remainder-sets remain, or the bag value raises above $\ell$.
\item If there are still some unused remainder-sets and high-value goods, let $a := a+1$ and go back to Step 2.
\end{enumerate}

The bag-filling stops when either there are no more high-value goods, or no more remainder-sets.
In the former case, Alice has all $n-k'$ required bundles
($s$ from Step 1 and $n-k'-s$ from Step 2), and the construction is done. 
We now analyze the latter case.

By construction, we go to the next bag only after the current bag becomes at least $\ell$. Therefore, all bags except the last one are valued at least $\ell$. Our goal now is to prove that the number of these ``all bags except the last one'' is sufficiently large.

Let $t$ be the number of bundles constructed with a value of at least $\ell$.
For each $a\in[t]$, 
The $a$-th bag contains the high-value goods
$g_{j[a,1]},\ldots,g_{j[a,\ell]}$
and some remainder-sets.
How much remainder-sets should it contain?
Suppose it contains remainder-sets with a total value of $\sum_{l=1}^\ell r_{j[a,l]}$.
Then, the total bundle value is 
$\sum_{l=1}^\ell v(M_{j[a,l]})$. 
By assumption, the total value of every $\ell$ MMS bundles is at least $\ell$, so the bundle value is at least $\ell$.
Therefore, to make bundle $a$ acceptable, it is sufficient to add to it a value of $\sum_{l=1}^\ell r_{j[a,l]}$.

Denote by $j[a,*]$ the index of the last remainder-set added to bag $a$ (bringing its value from less-than-$\ell$ to at-least-$\ell$).
The total value of remainder-sets in the bag is thus 
less than $r_{j[a,*]} + \sum_{l=1}^\ell r_{j[a,l]}$. 

The total value of remainder-sets in the unfilled $(t+1)$-th bag is less than 
$\sum_{l=1}^\ell r_{j[t+1,l]}$,
where $j[t+1,1],\ldots,j[t+1,\ell]$ are indices  of some remaining high-value goods.
Therefore, the total value of remainder-sets 
in all $t+1$ bags together 
satisfies 
\begin{align}
\notag
v
\left(
\bigcup_{j  \in P_+ \cup P_-} R'^P_{j}
~~\cup~~ 
\bigcup_{c=1}^{k} \bigcup_{l=1}^{\ell} R'^U_{c,l}
\right) 
&~~<~~
\sum_{a=1}^t \left( 
r_{j[a,*]} 
+
\sum_{l=1}^\ell r_{j[a,l]}
\right)
~~+~~ 
\left(
\sum_{l=1}^\ell r_{j[t+1,l]}
\right)
\\
\label{eq:bagvalue}
&=
\left(\sum_{a=1}^{t+1} 
\sum_{l=1}^\ell r_{j[a,l]}
\right)
~~+~~ \left(\sum_{a=1}^t 
r_{j[a,*]}
 \right).
\end{align}

Combining \eqref{eq:total-value-remainders} and \eqref{eq:bagvalue} gives
\begin{align}
\left(\sum_{a=1}^{t+1} 
\sum_{l=1}^\ell r_{j[a,l]}
\right)
+ \left(\sum_{a=1}^t 
r_{j[a,*]}
 \right)
>
\sum_{j  \in P_+ \cup P_-} r^P_{j}
.
\end{align}
In the left-hand side there are 
$\ell(t+1)+t = (\ell+1)t+\ell$ 
terms, 
while in the right-hand side there are 
$(\ell+1)s$ terms --- $\ell+1$ for each bundle constructed in Step 1.
We now show that each term in the left-hand side is equal or smaller than a unique term in the right-hand side. Since the left-hand side is overall larger than the right-hand side, this indicates that the left-hand side must have more terms, that is, 
$(\ell+1)t+\ell > (\ell+1)s$.
This implies that $t\geq s$, i.e., Alice has successfully constructed from the remainder-sets some
$s$
acceptable bundles.
\begin{itemize}
\item Consider first the $\ell (t+1)$ terms $r_{j[a,l]}$, and compare them to $r^P_{j}$ for $j \in P_+$. 
Since the bundles in Step 1 were constructed in ascending order of value, starting at the lowest-valued available goods in 
each of $G^n_1\ldots,G^n_{\ell}$,
every index $j[a,l]$ is smaller than any index $j \in G^n_l$.
Therefore, every term $r_{j[a,l]}$ is smaller than some unique term $r^P_{j}$ for $j \in G^n_l$, for every $l\in[\ell]$.
\item Consider now the $t$ terms $r_{j[a,*]}$, and compare them  to $r^P_{j}$ for $j \in P_-$.
Each of the indices $j[a,*]$ is an index of some unique remainder-set, so it is either equal to some unique index $j\in P_+\cup P_-$, or to some unique index $_{c,l}$ (the index some remainder-set $R^U_{c,l}$ of some unacceptable bundle $B_c$). All indices $_{c,l}$ are in $\{1,\ldots, \ell n\}$, so they are smaller than the indices $j \in P_-$
.
Therefore, every $r_{j[a,*]}$ is either equal or smaller than some unique term $r^P_{j}$ for $j\in P_-$.
\end{itemize}
So Alice has $s$ new acceptable bundles.
The waste of each of these is $r_{j[a,*]}$, which --- as mentioned above --- is equal to or smaller than some unique term $r^P_{j}$ for $j\in P_-$.
Therefore, the total waste of all these $s$ bundles is at most the following sum of $s$ terms: $\sum_{j \in P_-} r^P_{j}$.

Recall that the waste of each of the $s$ acceptable bundles from Step 1 was at most $v(g_{j})$ for some $j\in P_-$. Therefore, the total waste of the $2 s$ acceptable bundles constructed so far is at most 
\begin{align*}
&\sum_{j \in P_-} r^P_{j} + \sum_{j \in P_-} v(g_{j})
\\
=&
\sum_{j \in P_-} (r^P_{j} + v(g_{j}))
\\
=&
\sum_{j \in P_-} v(M_{j})
\\
\leq &
\sum_{j \in P_-} (\ell-x) && \text{by the normalization (Section~\ref{sub:normalize-goods})}
\\
=&
|P_-|\cdot (\ell-x)
\\
=&
s\cdot (\ell-x).
\end{align*}
Therefore, the average waste per bundle is at most $s(\ell-x)/(2s) = (\ell-x)/2$.

\subsection{Step 3: Plain Bag-filling.}
At this stage, there are no more high-value goods outside $H_+$. Therefore, Alice can construct the remaining bundles by plain bag-filling, initializing each bag with some $\ell$-tuple of unused goods remaining in $H_+$, and filling it with some low-value goods outside $H_+$.
Since the waste in each bundle is at most $(\ell-x)/2$, Lemma \ref{lem:waste2} implies that the total number of constructed bundles is $n-k'$.

This completes the proof that  
$\ell$ is a reasonable threshold for Algorithm \ref{alg:(L+1/2)n}.
Therefore, the algorithm finds the allocation promised in Theorem \ref{thm:l-out-of-d-existence}. 

\subsection{Limits of Algorithm \ref{alg:(L+1/2)n}}
To illustrate the limitation of Algorithm \ref{alg:(L+1/2)n},
we show that
it cannot guarantee 
$1$-out-of-$((\ell+\frac{1}{2}) n -2)$ 
MMS. For simplicity we assume that $n$ is even so that $(\ell+\frac{1}{2})n$ is an integer.
\begin{example}[\textbf{Tight bound for our technique}]
\label{exm:lone-divider-2}
Suppose that in the first iteration all agents except the divider have the following MMS bundles:
\begin{itemize}
\item $\ell n-1$ bundles are made of two goods with values $1-\varepsilon$ and $\varepsilon$.
\item One bundle is made of two goods with values $1-\ell n\varepsilon$ and $\ell n\varepsilon$.
\item $n/2-2$ bundles are made of two goods with values $1/2,1/2$.
\end{itemize}
So their  $1$-out-of-$(\ell n + n/2-2)$ MMS equals $\ell$.
However, it is possible that the first divider takes an unacceptable bundle containing 
$\ell-1$ goods of value $1-\varepsilon$,
the good of value $1-\ell n\varepsilon$, 
and the $\ell n-1$ goods of value $\varepsilon$.
Note that this bundle is $\ell$-balanced.
All remaining goods have a value of less than $1$,
so an acceptable bundle requires at least $\ell+1$ goods.
However, the number of remaining goods is only
$\ell n - \ell + 1 + n - 4 = (\ell+1)(n-1) -  2 $:
$\ell n-\ell$ goods of value $1-\varepsilon$,
one good of value $\ell n\varepsilon$ and $n-4$ goods of value $1/2$. Hence, at most $n-2$ acceptable bundles can be constructed. $\blacksquare$
\end{example}

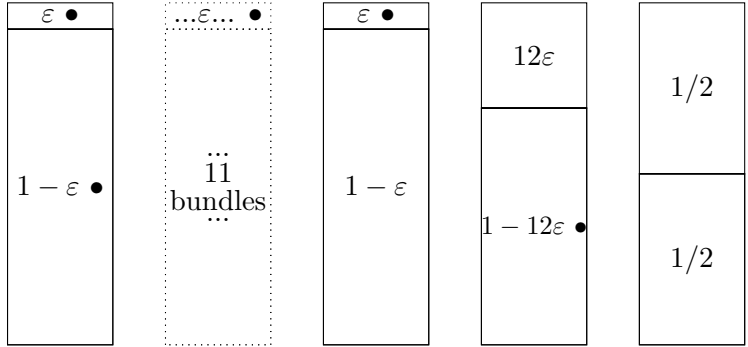
\begin{figure}
\begin{center}
\begin{tikzpicture}[scale=0.7]
\mmsbundle{0}{6}{$1-\varepsilon$ $\bullet$}{6.5}{$\varepsilon$ $\bullet$}

\draw [dotted] (3,0) rectangle +(2,6)
node[pos=.5] {\shortstack{$...$\\$11$\\bundles\\$...$}};
\draw [dotted] (3,0) rectangle +(2,6)
rectangle +(0,6.5) node[pos=.5] {...$\varepsilon$... $\bullet$};

\mmsbundle{6}{6}{$1-\varepsilon$}{6.5}{$\varepsilon$ $\bullet$}

\mmsbundle{9}{4.5}{\small$1-12 \varepsilon$ $\bullet$}{6.5}{$12 \varepsilon$}

\mmsbundle{12}{3.25}{$1/2$}{6.5}{$1/2$}

 \end{tikzpicture}
\end{center}
\caption{
\label{fig:lone-divider-2}
An illustration of the goods' values in Example \ref{exm:lone-divider-2}, for $n=6$ and $\ell=2$.
Here, $(\ell+1/2)n-2 = 13$. The $2$-out-of-$13$ MMS is $2$.
Each rectangle represents an MMS bundle containing two goods.
The first divider takes the goods marked by a bullet. Note that this is a 2-balanced bundle. The next divider cannot construct $5$ bundles of value 2 from the remaining goods.
}
\end{figure}

\section{Ordinal Approximation for Goods in Polynomial Time}
\label{sec:goods-poly}
Algorithm \ref{alg:(L+1/2)n} guarantees that each agent receives an $\ell$-out-of-$d$ MMS allocation for $d \geq \goodsapp$. However, the algorithm requires exact MMS values to determine whether a given bundle is acceptable to each agent. Since computing an exact MMS value for each agent is NP-hard, Algorithm \ref{alg:(L+1/2)n} does not run in polynomial-time even for the case of $\ell = 1$.
The objective of this section is to develop  polynomial-time approximation algorithms for computing $\ell$-out-of-$d$ MMS allocations.

We utilize optimization techniques used in the \emph{bin covering} problem. This problem was presented by \citet{assmann1984dual} as a dual of the more famous \emph{bin packing} problem. 
In the bin covering problem, the goal is to fill bins with items of different sizes, such that the sum of sizes in each bin is at least $1$, and subject to this, the number of bins is maximized. This problem is NP-hard, but several approximation algorithms are known.
These approximation algorithms typically accept a bin-covering instance $I$ as an input and fill at least
 $a\cdot (OPT(I) - b)$ bins, 
where $a<1$ and $b>0$ are constants,
and $OPT(I)$ is the maximum possible number of bins in $I$.
Such an algorithm can be used directly to find an ordinal approximation of an MMS allocation when all agents have \emph{identical} valuations. Our challenge is to adapt them to agents with \emph{different} valuations.

\subsection{The case when $\ell=1$}
\label{sub:l=1}

\begin{algorithm}[t]
\caption{
\label{alg:3n/2-polytime}
Bidirectional bag-filling
}
\begin{algorithmic}[1]
\REQUIRE An instance $\ins{N, M, V}$ 
and threshold values $(t_i)_{i=1}^n$.

\ENSURE At most $n$ subsets $A_i$ satisfying $v_{i}(A_{i}) \geq t_{i}$.


\STATE  Order the instance in \textbf{descending order} of value as in Section \ref{sub:ordering}, so that for each agent $i$,
$v_i(g_1)\geq \cdots \geq v_i(g_m)$.


\FOR{$k = 1,2,\ldots$:}
\STATE 
\label{step:new-bag}
Initialize a bag with the good $g_k$. 
\STATE 
Add to the bag zero or more remaining goods in \textbf{ascending order} of value, 
until at least one agent $i$ values the bag at least $t_i$.
\STATE 
Give the goods in the bag to an arbitrary agent $i$ who values it at least $t_i$.
\STATE  If every remaining agent $i$ values the remaining goods at less than $t_i$, stop.
\ENDFOR
\end{algorithmic}
\end{algorithm}

For the case when $\ell=1$, we adapt the algorithm of \citet{csirik1999two}, which finds a covering with at least $\frac{2}{3}\cdot (OPT(I)-1)$ bins (an approximation with $a=\frac{2}{3}$ and $b=1$).
Algorithm \ref{alg:3n/2-polytime} generalizes the aforementioned algorithm to MMS allocation of goods.
Thus, the algorithm of \citet{csirik1999two} corresponds to a special case of Algorithm \ref{alg:3n/2-polytime} wherein
\begin{itemize}
\item  All agents have the same $v_i$ (describing the item sizes); and
\item  All agents have the same $t_i$ (describing the bin size).\footnote{
There is a minor difference: 
we initialize the first bag with only a single good from the left ($g_1$) before filling it with goods from the right ($g_m,g_{m-1},\ldots$).
In contrast, \citet{csirik1999two} fill the first bag with several goods from the left ($g_1,g_2,\ldots$ while its value is less than the bin size),
and only then start filling it with goods from the right.
However, this difference is not substantial: their proof of the approximation ratio assumes only that each bin has at least one good from the left and one good from the right, so the same proof holds for our variant.
}
\end{itemize}

For this case, we have the following lemma:
\begin{lemma}[Lemma 4 of \citet{csirik1999two}]
\label{lem:csirik}
When all agents have the same valuation $v$ and the same threshold $t$, Algorithm \ref{alg:3n/2-polytime} allocates at least $\frac{2}{3}(OPT(v,t)-1)$ bundles, where $OPT(v,t)$ is the maximum number that can be filled.
\end{lemma}

Note that Algorithm \ref{alg:3n/2-polytime} works for any selection of the threshold values $t_i$, but if the thresholds are too high, it might allocate fewer than $n$ bundles. Our challenge now is to compute thresholds for which $n$ bundles are allocated. 
To compute a threshold for agent $i$, we simulate Algorithm \ref{alg:3n/2-polytime} using $n$ clones of $i$, that is, $n$ agents with  valuation $v_i$. We look for the largest threshold for which this simulation allocates at least $n$ bundles.
\begin{definition}
The \emph{1-out-of-$n$ bidirectional-bag-filling-share of agent $i$}, denoted $\BBFSA{n}_i$, is the largest value $t_i$ for which
Algorithm \ref{alg:3n/2-polytime} allocates at least $n$ bundles when executed with 
$n$ agents with identical valuation $v_i$ and identical threshold $t_i$.
\end{definition}
\label{def:bbfs}
The BBFS of agent $i$ can be computed using binary search up to $\varepsilon$, where $\varepsilon$ is the smallest difference between values that is allowed by their binary representation.
The following lemma relates the BBFS to the MMS.
\begin{lemma}
\label{lem:bbfs}
For any integer $n\geq 1$ and agent $i\in[n]$,
\begin{align*}
\BBFSA{n}_i
\geq
\mms{i}{1}{\ceil{\frac{3}{2}n}}{M}.
\end{align*}
\end{lemma}
\begin{proof}
Let $t_i := \mms{i}{1}{\ceil{\frac{3}{2}n}}{M}$.
By definition of MMS, there is a partition of $M$ into $\ceil{\frac{3}{2}n}$ bundles of size at least $t_i$.
By Lemma \ref{lem:csirik},
the Bidirectional-Bag-Filling algorithm
with valuation $v_i$ and bin-size $t_i$
fills at least 
$\frac{2}{3}(\ceil{\frac{3}{2}n}-1)$
bundles, which means at least $n$ bundles since the number of bundles is an integer.
By definition of the BBFS, since Algorithm \ref{alg:3n/2-polytime} allocates at least $n$ bundles with threshold $t_i$, we have $t_i\leq \BBFSA{n}_i$.

%
\end{proof}

We define an allocation as \emph{BBFS-fair} if it allocates to each agent $i\in[n]$ a bundle with a value of at least $\BBFSA{n}_i$.
Lemma \ref{lem:bbfs} indicates that a BBFS-fair allocation is also 1-out-of-$\ceil{3n/2}$ MMS-fair, though the BBFS may be larger than
1-out-of-$\ceil{3n/2}$ MMS.

\begin{lemma}
\label{lem:bbsfair}
A BBFS-fair allocation always exists, and can be found in time polynomial in the length of the binary representation of the problem.
\end{lemma}
\begin{proof}
We first show that, when
Algorithm \ref{alg:3n/2-polytime}
is executed with threshold values 
$t_i = \BBFSA{n}_i$
for all $i\in[n]$,
it allocates $n$ bundles.
For each $j\geq 1$, denote:
\begin{itemize}
\item $A_j$ --- the bundle allocated at iteration $j$ of Algorithm \ref{alg:3n/2-polytime} with the true (different) valuations $v_1,\ldots,v_n$.
\item $B^i_j$ --- 
the bundle allocated at iteration $j$ of agent $i$'s successful simulation with threshold $t_i = \BBFSA{n}_i$.
\end{itemize}
We claim that, for every $k\geq 1$,
the set of goods allocated before step $k$ by the global algorithm is a subset of the goods allocated before step $k$ during agent $i$'s simulation, .
That is, 
$\bigcup_{j=1}^{k-1} A_j \subseteq \bigcup_{j=1}^{k-1} B^i_j$ for any remaining agent $i$.

The claim is proved by induction on $k$. 
The base is $k=1$. Before step $1$, both $\bigcup_{j=1}^{k-1} A_j$ and $\bigcup_{j=1}^{k-1} B^i_j$ are empty, so the claim holds vacuously.
Let $k\geq 1$. We assume the claim is true before iteration $k$, and prove that it is still true after iteration $k$.
The initial goods $g_1,\ldots,g_k$ are obviously allocated in both runs.
In agent $i$'s simulation, some additional goods $g_m,\ldots,g_s$ are allocated, for some $s\leq m$;
in the global run, goods $g_m,\ldots,g_r$ are allocated, for some $r\leq m$.
The induction assumption implies that $r \geq s$
(weakly fewer goods are allocated in the global run).
In iteration $k$, both runs initialize the bag with the same good $g_k$.
In $i$'s simulation, the bag is then filled with 
goods $g_{s-1},\ldots,g_{s'}$ for some $s'<s$,
such that $v_i(\{g_k,g_{s-1},\ldots,g_{s'}\})\geq t_i$.
In the global run, the bag is filled with 
goods $g_{r-1},\ldots,g_{r'}$ for some $r'<r$.
It is sufficient to prove that $r'\geq s'$.
Indeed, if no agent takes the bag until it contains the goods $(\{g_k,g_{r-1},\ldots,g_{s'}\})$,
then because $r\geq s$, the bag value is at least
$v_i(\{g_k,g_{s-1},\ldots,g_{s'}\})\geq t_i$.
Therefore, it is acceptable to agent $i$, so the algorithm allocates it (either to $i$ or to another agent).

This completes the proof of the claim.
The claim implies that, as long as $k < n$, 
the goods in $B^i_n$ are still available. This means that 
agent $i$ values the remaining goods at least $t_i$.
This is true for every remaining agent; therefore, the global algorithm continues to run until it allocates $n$ bundles.

The binary search and the simulation runs for each agent $i$ take time polynomial in the length of the binary representation of the valuations.
Once the thresholds are computed, Algorithm \ref{alg:3n/2-polytime} obviously runs in polynomial time.
This completes the proof of the lemma.
\end{proof}

Lemmas \ref{lem:bbfs} and \ref{lem:bbsfair} together imply:
\begin{theorem} \label{thm:3n-2poly}
There is an algorithm that computes a 
1-out-of-$\ceil{3n/2}$ MMS allocation
in time polynomial in the length of the binary representation of the problem.
\qed
\end{theorem}

\begin{example}[\textbf{Computing thresholds}]
Consider a setting with $m=6$ goods and $n=3$ agents with the following valuations:
\begin{center}
\begin{tabular}{c|cccccc|c}
    & $g_{1}$ & $g_{2}$ & $g_{3}$ & $g_{4}$ & $g_{5}$ & $g_{6}$ & $t_{i}$ \\\hline
    $v_{1}$ & \circled{10} & 8 & 6 & 3 & 2 & 1 & 9\\
    $v_{2}$ & 12 & \circled{7} & 6 & 5 & \circled{4} & \circled{2} & 11\\
    $v_{3}$ & 9 & 8 & \circled{7} & \circled{4} & 3 & 1 & 10
\end{tabular}
\end{center}

Each player computes a threshold via binary search on $[0, v_{i}(M)]$ for the maximum value $t_{i}$ such that the simulation of Algorithm \ref{alg:3n/2-polytime} yields three bundles. For agent $1$, the simulation with $t_{1} = 9$ yields bundles $\{g_{1}\}, \{g_{2}, g_{6}\}, \{g_{3}, g_{4}, g_{5}\}$. The corresponding simulation with $t_{1} = 10$ yields bundles $\{g_{1}\}, \{g_{2}, g_{5}, g_{6}\}$ with $\{g_{3}, g_{4}\}$ insufficient to fill a third bundle.

After all thresholds have been determined from simulations, Algorithm \ref{alg:3n/2-polytime} computes the circled allocation. Theorem \ref{thm:3n-2poly} guarantees that this allocation is at least $1$-out-of-$5$ MMS. Here the circled allocation satisfies $1$-out-of-$3$ MMS. $\blacksquare$
\end{example}

\begin{remark}
When $n$ is odd, there is a gap of $1$ between the existence result for
1-out-of-$\floor{3n/2}$ MMS,
and the polynomial-time computation result
for
1-out-of-$\ceil{3n/2}$ MMS.
\end{remark}

In experimental simulations on instances generated uniformly at random, Algorithm \ref{alg:3n/2-polytime} significantly outperforms the theoretical guarantee of 1-out-of-$\ceil{3n/2}$ MMS. In Appendix~\ref{sec:bidi-vs-uni}, we provide detailed experimentations and compare the bidirectional bag-filling algorithm with other bag-filling methods (e.g. the unidirectional bag-filling algorithm). 


\subsection{The case when $\ell>1$}
\label{sub:l>1}
So far, we could not adapt Algorithm \ref{alg:3n/2-polytime} to finding an $\ell$-out-of-$\floor{(\ell+1/2)n}$ MMS allocation for $\ell \geq 2$.
Below, we present a weaker approximation to MMS, based on the following lemma.
\begin{lemma}
\label{lem:L-out-of-D}
For all integers $d > \ell\geq 1$:
\begin{align*}
\ell\cdot \MMSA{1}{d}(M)
\leq
\MMSA{\ell}{d}(M)
\leq
\ell\cdot \MMSA{1}{(d-\ell+1)}(M).
\end{align*}
\end{lemma}
\begin{proof}
For the leftmost inequality,
Let $A_1,\ldots,A_d$ be the optimal $d$-partition in the definition of $\MMSA{1}{d}(M)$, and suppose w.l.o.g. that the bundles are ordered by ascending value. Then:
\begin{align*}
\ell\cdot \MMSA{1}{d}(M) &= \ell\cdot v_i(A_1)
\\
&\leq v_i(A_1)+\cdots + v_i(A_{\ell})
\\
&\leq \MMSA{\ell}{d}(M),
\end{align*}
where the last inequality follows from the existence of a $d$-partition in which the $\ell$ least-valuable bundles are $A_1,\ldots,A_{\ell}$.

For the rightmost inequality,
let $B_1,\ldots,B_d$ be the optimal $d$-partition in the definition of $\MMSA{\ell}{d}(M)$, and suppose w.l.o.g. that the bundles are ordered by ascending value. Then:
\begin{align*}
\MMSA{\ell}{d}(M) &= v_i(B_1)+\cdots + v_i(B_{\ell})
\\
&\leq \ell \cdot v_i(B_{\ell})
\\
&\leq \ell\cdot \MMSA{1}{(d-\ell+1)}(M),
\end{align*}
where the last inequality is proved by the partition with $(d-\ell+1)$ bundles: $B_1\cup\cdots\cup B_{\ell}, B_{\ell+1},\ldots, B_{d}$, in which the value of each bundle is at least $v_i(B_{\ell})$.
\end{proof}

For any positive integer $d$, we can approximate $\MMSA{1}{d}(M)$ by using an approximation algorithm for bin-covering, which we call \emph{Algorithm JS} \citep{jansen2003asymptotic}.
\begin{lemma}[\citet{jansen2003asymptotic}]
\label{lem:js1}
For any $\varepsilon >0$, 
Algorithm JS 
 runs in time 
 $\widetilde{O}\left(
 \frac{1}{\varepsilon^6}m^2
 +
 \frac{1}{\varepsilon^{8.76}}
 \right)$.%
 \footnote{
A more exact expression for the run-time is
 $O\left(
\frac{1}{\varepsilon^5}
\cdot \ln{\frac{m}{\varepsilon}}
\cdot \max{(m^2,\frac{1}{\varepsilon}\ln\ln\frac{1}{\varepsilon^3})}
+
\frac{1}{\varepsilon^4}\mathcal{T_M}(\frac{1}{\varepsilon^2})
\right)$, where $\mathcal{T_M}(x)$ is the run-time complexity of the best available algorithm for matrix inversion, which is currently $O(x^{2.38})$.
We simplified it a bit for clarity, and used $\widetilde{O}$ to hide the logarithmic factors.
}
If the sum of all valuations is at least $13t/\varepsilon^3$ (where $t$ is the bin size), then Algorithm JS fills at least $(1 - \varepsilon)\cdot \mathrm{OPT}(I) - 1$ bins.
\end{lemma}
We can choose $\varepsilon$ based on the instance, and get the following simpler guarantee.
\begin{lemma}
\label{lem:js2}
Algorithm JS fills at least
\begin{align*}
\mathrm{OPT} - 2.35\cdot  \mathrm{OPT}^{2/3} - 1
\end{align*}
bins, and runs in time
$\widetilde{O}(m^{4})$.
\end{lemma}
\begin{proof}
If any input value is at least $t$, then it can be put in a bin of its own, and this is obviously optimal. So we can assume w.l.o.g. that all input values are smaller than $t$.

Let $s$ be the sum of values, and set $\varepsilon := (13 t / s)^{1/3}$.
The number of bins in any legal packing is at most $s/t$, so 
\begin{align*}
\mathrm{OPT} \leq& s/t
\\
t/s \leq  &1/\mathrm{OPT}
\\
\varepsilon \leq  &(13/\mathrm{OPT})^{1/3}
\\
&\approx 2.35/ \mathrm{OPT}^{1/3}.
\end{align*}


The $\varepsilon$ is chosen such that $s = 13t/\varepsilon^3$. So by Lemma \ref{lem:js1}, the number of bins filled by Algorithm JS is at least 
\begin{align*}
&
\mathrm{OPT} - \varepsilon\cdot \mathrm{OPT} - 1
\\
\geq&
\mathrm{OPT} - 2.35\cdot  \mathrm{OPT}^{2/3} - 1.
\end{align*}

Since by assumption each value is smaller than $t$, we have $s<m t$, so 
$\varepsilon  > (13/m)^{1/3}$ 
and
$1/\varepsilon  \in O(m^{1/3})$.  Therefore, the run-time is in
\begin{align*}
&
\widetilde{O}\left(
 \frac{1}{\varepsilon^6}m^2
 +
 \frac{1}{\varepsilon^{8.8}}
\right)
\\
\approx&
\widetilde{O}\left(
 m^2\cdot m^2
 +
m^{2.92}
\right)
\\
\approx&
\widetilde{O}\left(m^4\right).
\qedhere
\end{align*}
%
%
%
%
%
%
\end{proof}

Analogously to the definition of the BBF-Share (BBFS) (Definition \ref{def:bbfs}), 
we define the 1-out-of-$d$ \emph{JS-Share} (JSS) of each agent $i$  as
the largest value $t_i$ for which
Algorithm JS fills at least $d$ bins when executed with 
valuation $v_i$ and bin-size $t_i$.
The JSS can be computed up to any desired accuracy using binary search.

Clearly, 
$
\mms{i}{1}{d}{M}
\geq
\JSSA{d}_i
$.
Analogously to Lemma \ref{lem:bbfs}, we have
\begin{lemma}
\label{lem:jss}
For any integer $d\geq 1$ and agent $i\in[n]$:
\begin{align*}
\JSSA{d}_i
\geq
\mms{i}{1}{\ceil{d+15\cdot d^{2/3}+1}}{M}.
\end{align*}
\end{lemma}
\begin{proof}
Let $t_i := \mms{i}{1}{\ceil{d+15\cdot d^{2/3}+1}}{M}$.
By definition of MMS, there is a partition of $M$ into $\ceil{d+15\cdot d^{2/3}+1}$ bundles of size at least $t_i$.
By Lemma \ref{lem:js2}, Algorithm JS with bin-size $t_i$ fills at least
\begin{align*}
&
(d+15\cdot d^{2/3}+1) - 2.35\cdot (d+15\cdot d^{2/3}+1)^{2/3} - 1
\\
\geq
&
(d+15\cdot d^{2/3}+1) - 2.35\cdot (16 d)^{2/3} - 1
\\
\geq
&
(d+15\cdot d^{2/3}+1) - 14.92 d^{2/3} - 1
\\
\geq
&
d
\end{align*}
bins.
By definition of the JSS, since Algorithm JS allocates at least $d$ bins with size $t_i$, we have $t_i\leq \JSSA{d}_i$.
\end{proof}

Now, each agent can participate in the algorithm of Section \ref{sec:goods-lone} without computing the exact MMS value. Given an integer $\ell\geq 2$,
let $d := \floor{(\ell+\frac{1}{2})n}$.
Each agent $i$ can compute
the value of $\JSSA{d}_i$
using binary search.
The search also finds a partition of $M$ into $d$ bundles, each of which has a value of at least $\JSSA{d}_i$.
The agent can now use this partition for scaling:
the valuations are scaled such that the value of each bundle in the partition is exactly $1$.
The algorithm in Section \ref{alg:(L+1/2)n} guarantees to each agent a bundle with value at least $\ell$, which is at least
$\ell\cdot \JSSA{d}_i$. 
By Lemma \ref{lem:jss},
this value is at least
$\ell\cdot\mms{i}{1}{\ceil{d+15d^{2/3}+1}}{M}$.
By the right-hand side of Lemma \ref{lem:L-out-of-D}, it is at least 
$\mms{i}{\ell}{\ceil{d+15d^{2/3}+\ell}}{M}$.
Thus, we have proved the following theorem.

\begin{theorem}
\label{thm:goods-approx}
Let 
$\ell\geq 2$ an integer, and 
$d := \floor{(\ell+\frac{1}{2})n}$.
It is possible to compute an allocation in which the value of each agent $i$ is at least 
\begin{align*}
\mms{i}{\ell}{\ceil{d+15d^{2/3}+\ell}}{M},
\end{align*}
in time 
 $\widetilde{O}\left(
 n\cdot m^{4}
\right)$.
\end{theorem}

\section{Future Directions}
The existence of tighter ordinal approximations that improve $\ell$-out-of-$\floor{(\ell+1/2)n}$ MMS allocations is a compelling open problem. 
Specifically, one can generalize the open problem raised by  \citet{budish2011combinatorial} and ask, for any $\ell\geq 1$ and $n\geq 2$: does there exist an $\ell$-out-of-$(\ell n + 1)$ MMS allocation? 

For the polynomial-time algorithm when $\ell = 1$, 
we extend the bin covering algorithm of \citet{csirik1999two}.
We believe that the interaction between this problem and fair allocation of goods may be of independent interest, as it may open new ways for developing improved algorithms. For example,  \citet{csirik1999two} also present a $3/4$ approximation algorithm for bin covering, which may potentially be adapted to yield a  $1$-out-of-$\ceil{4n/3}$ MMS allocation.
Similarly, \citet{csirik2001better} and \citet{jansen2003asymptotic} present polynomial-time approximation schemes for bin covering, which may yield even better MMS approximations in future work.

Finally, it is interesting to study ordinal maximin approximation for items with non-positive valuations (i.e. chores), as well as for mixtures of goods and chores. Techniques for allocation of goods do not immediately translate to achieving approximations of MMS when allocating chores, so new techniques are needed
\citep{hosseini2022ordinal}.

\section*{Acknowledgments}
Hadi Hosseini acknowledges support from NSF IIS grants \#2052488 and \#2107173.
Erel Segal-Halevi is supported by the Israel Science Foundation (grant no. 712/20).

We are grateful to Thomas Rothvoss, Ariel Procaccia,
Joshua Lin,
Inuyasha Yagami,
Chandra Chekuri, Neal Young,
and the anonymous referees of EC 2021 and JAIR for their valuable feedback.

\newpage
\appendix
\section*{APPENDIX}

\section{Comparing Ordinal and Multiplicative Approximations}
\label{sec:simulation}
Our ordinal guarantees may be better than the best known multiplicative MMS approximation (i.e. $3/4$) when the number of goods is large compared to the number of agents.

To illustrate, consider the extreme case in which there are infinitely many goods of equal value
(alternatively, suppose there are infinitely many goods with values that are independent and identically-distributed random variables).
Then 1-out-of-$n$ MMS converges to $1/n$
(with probability $1$, by the law of large numbers).%
\footnote{
An accurate computation of the convergence rate of the MMS to $1/n$ is beyond the scope of the present paper.
We refer the interested reader to 
\citet{mertens2001physicist}, 
who studies a closely-related problem:
the probability distribution of the smallest \emph{difference} between the highest-valued bundle and the lowest-valued bundle in an $n$-partition.
}
The $\ell$-out-of-$(\ell+\frac{1}{2})n$ MMS converges to
$2 \ell/(2 \ell+1)$ of this value, which is larger than $3/4+1/(12n)$ for $\ell\geq 2$, and approaches $1$ when $\ell\to \infty$.

In this section, we present a simple simulation experiment that compares the value of the $\ell$-out-of-$d$ MMS guaranteed by Theorem \ref{thm:l-out-of-d-existence}
(where $d=\floor{(\ell+1/2)n}$) with the best known multiplicative approximation of $1$-out-of-$n$ MMS, which is $\frac{3}{4} + \frac{1}{12n}$ \citep{garg2020improved}.
Our results show that the ordinal approximation for $\ell\geq 2$ is better than the multiplicative approximation already for $m \approx 20 n$, when the values are sampled from some natural distributions. 
We note that the simulations only compare the worst-case guarantees and not the actual algorithm performance.

Since computing the exact MMS is NP-hard,%
\footnote{
Using integer linear programming, we could compute the exact value of 1-out-of-$4$ MMS for $m > 200$ goods in reasonable time. However, we were not able to scale our computations for larger values of $n$.
}
we used a lower bound for our ordinal approximation and an upper bound for the ``competition'', so that our ratio is a lower bound for the real ratio.
For our ordinal approximation, we computed a lower bound using the \emph{greedy number partitioning} algorithm 
\cite{graham1966bounds,graham1969bounds}.
This algorithm is known to attain a reasonable approximation of the maximin share both in the worst case \citep{deuermeyer1982scheduling,csirik1992exact} and in the average case \citep{frenk1986rate}.
Given an integer $d$,
the algorithm initializes $d$ empty bundles. 
It iterates over the goods in descending order of their value, and puts the next good in the bundle with the smallest total value so far (breaking ties arbitrarily).
Once all goods are allocated, the sum of values in the $\ell$ bundles with the smallest values is a lower bound for the $\ell$-out-of-$d$ MMS.
Taking instead the smallest value times $\ell$ 
(which approximates $\ell\cdot\MMSA{1}{d}$)
yields nearly identical results.
For the multiplicative approximation, we just use the proportional share $v_i(M)/n$ 
as an upper bound for agent $i$'s  $1$-out-of-$n$ MMS.

For various values of $m$, we chose $m$ random integers to use as the good values.
We performed three simulations, in which the values were
distributed (a) uniformly at random in $[1,1000]$, (b) uniformly at random in  $[1000, 2000]$, and (c) geometrically with mean value of $1000$.
We modified $n$ between $4$ and $20$, and $m$ between $4 n$ and $80 n$. The results for all $n$ were very similar. 
While our approximation for $\ell=1$ is generally worse than $3/4$ of the MMS, our approximation for $\ell\geq 2$ is better already for $m\approx 20 n$, and it becomes better as $m$ grows.
Figure \ref{fig:simulation} illustrates these observations.\footnote{
Source code for the experiments is available at \url{https://github.com/erelsgl/ordinal-maximin-share}.}

\begin{figure}
\begin{center}
\begin{subfigure}[b]{\textwidth}
\caption{Values distributed uniformly in $[1,\ldots,1000]$:}
\includegraphics[width=0.45\textwidth]{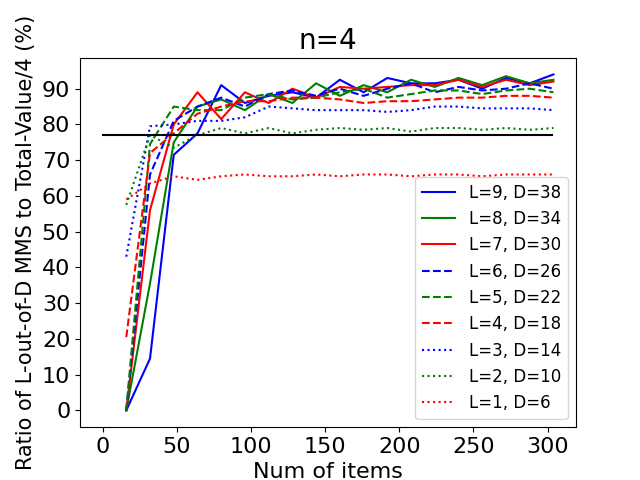}
~~
\includegraphics[width=0.45\textwidth]{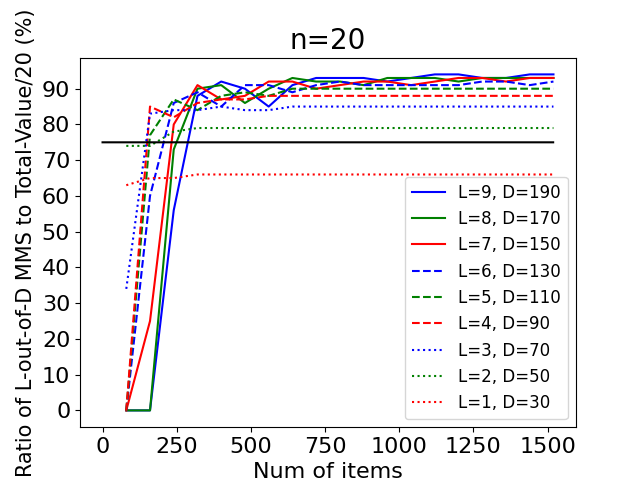}
\end{subfigure}
~\\
\begin{subfigure}[b]{\textwidth}
\caption{Values distributed uniformly in $[1000,\ldots,2000]$:}
\includegraphics[width=0.45\textwidth]{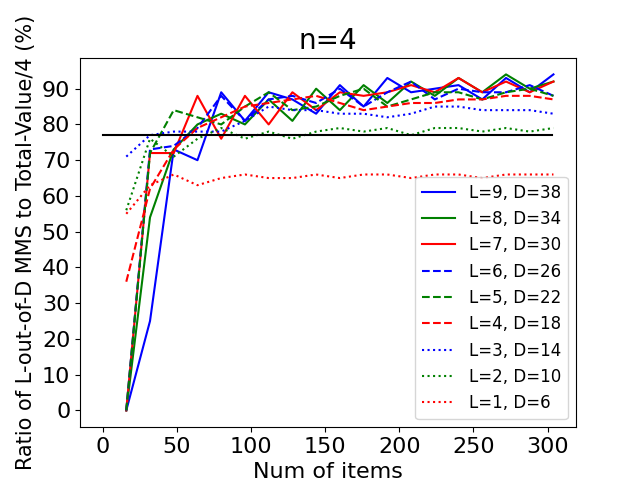}
~~
\includegraphics[width=0.45\textwidth]{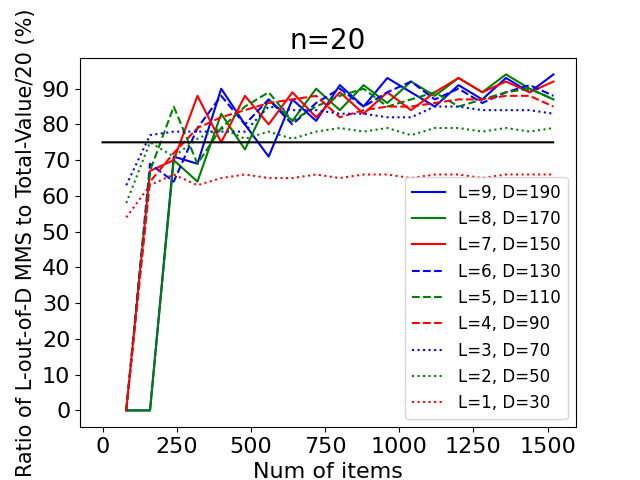}
\end{subfigure}
~\\
\begin{subfigure}[b]{\textwidth}
\caption{Values distributed geometrically with mean value $1000$:}
\includegraphics[width=0.45\textwidth]{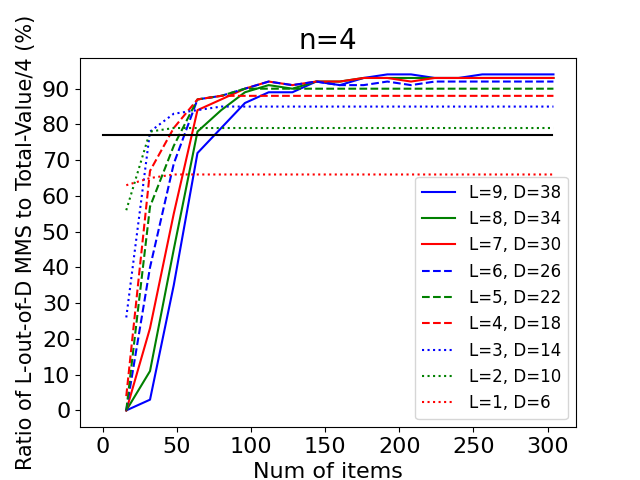}
~~
\includegraphics[width=0.45\textwidth]{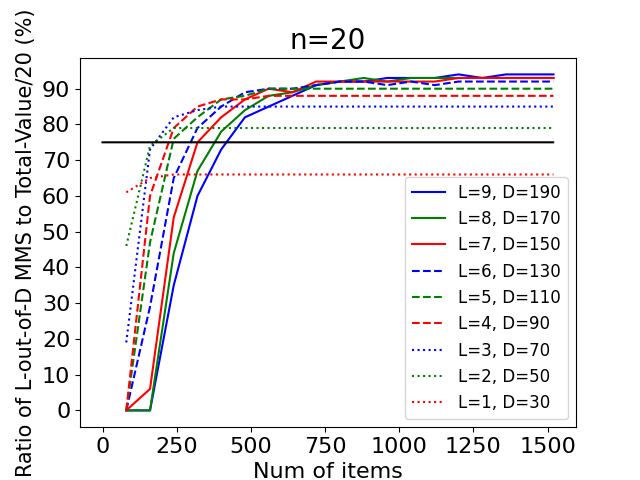}
\end{subfigure}
\end{center}
\caption{
\label{fig:simulation}
The value of the $\ell$-out-of-$d$ MMS (where $d=(\ell+1/2)n$) guaranteed by our Theorem \ref{thm:l-out-of-d-existence},
as a fraction of the $1$-out-of-$n$ MMS, for different values of $\ell$ and $n$ (at the left $n=4$ and at the right $n=20$).
The horizontal black line represents $3/4+1/(12n)$ of the $1$-out-of-$n$ MMS
}
\end{figure}
\section{Failure of Some Common Techniques for Approximate-MMS Allocation}
\label{sec:negative}
Our $1$-of-$\floor{3n/2}$ MMS guarantee
seems very similar to multiplicative $2/3$-MMS approximation,
as both can be seen as approximations of $2/3\cdot (v_i(M)/n)$.
One could expect that the same techniques should work in both cases. To illustrate that this is not the case, we consider one such technique, recently used by \citet{garg2018approximating} to find a $2/3$-MMS allocation in polynomial time. For completeness, we briefly describe their algorithm below (the detailed steps can be found in \citet{garg2018approximating}):
\begin{enumerate}
\item Scale the valuations such that all $n$ agents value the set of all goods at $3 n / 2$ (this implies that their MMS value is at most $3/2$).
\item Order the instance such that $v_i(g_1)\geq \cdots \geq v_i(g_m)$ to all agents $i$.
\item If an agent $i$ values $g_1$ by at least $1$, then allocate $g_1$ to $i$ and recurse with the remaining goods and $n-1$ agents.
\item If an agent $i$ values $\{g_n, g_{n+1}\}$ by at least $1$, 
then allocate $\{g_n, g_{n+1}\}$ to $i$ and recurse with the remaining goods and $n-1$ agents.
\item At this point, all agents value $g_1,\ldots,g_n$ at less than $1$, and $g_{n+1},\ldots,g_m$ at most $1/2$. Allocate the goods using bag-filling, initializing each bag $j\in[n]$ with the good $g_j$.
\end{enumerate}
They prove that the reductions in Steps (3) and (4) above do not change the MMS value of the remaining agents. 
While the former reduction is valid regardless of the number of MMS bundles, 
the latter reduction crucially depends on the fact that there are $n$ MMS bundles, which implies (by the pigeonhole principle) that at least one MMS bundle contains at least two goods from $\{g_1,\ldots,g_{n+1}\}$.
This no longer holds when there are $3n/2$ MMS bundles.
Therefore, allocating goods $\{g_n,g_{n+1}\}$ is no longer a valid reduction.
Moreover, if we scale the valuations as in Step (1), we may be unable to give each agent a value of at least $1$. 

\begin{example}
Suppose $n=20$, there are $30$ goods with value $1-\varepsilon$ and one good with value $30\varepsilon$.
The instance is already  normalized, since $v_i(M)=30=3n/2$ for all $i\in N$.
The value of all goods is less than $1$.
However, the value of $\{g_n,g_{n+1}\}$ is more than $1$, and if we allocate such pairs of goods to agents, at most $15$ agents will receive a bundle.
This example shows that the threshold of $2 v_i(M) / (3n)$ might be too high (too wasteful) for this problem. $\blacksquare$
\end{example}

Note that the challenging case in Section \ref{sub:high-value-goods} (Case \#3, for $\ell=1$) is exactly the case in which more than $n$ MMS bundles contain high-value goods, and each such good is contained in a unique MMS bundle.

\section{Bidirectional vs Unidirectional Bag-filling}
\label{sec:bidi-vs-uni}
Theorem~\ref{thm:3n-2poly} guarantees that the bidirectional bag-filling algorithm with a careful set up of thresholds (Algorithm~\ref{alg:3n/2-polytime}) will compute an allocation which satisfies at least 1-of-$(3n/2+1)$ MMS. We empirically evaluated this algorithm by selecting different ways to set agent's thresholds.
As a baseline, we use a simpler, unidirectional bag-filling algorithm, 
where the goods are put into bags in decreasing order of their value.
For each version we ran two sets of experiments: (i) one experiment where agents threshold values are computed individually, and (ii) one experiment where all agents' thresholds are a common fraction 
of their proportional share.

In each experiment, we generated 1,000 instances for each pair of $n \in [3, 20]$ and $m \in [n, 100]$\footnote{Observe that when $m < 3n/2+1$, the 1-of-$(3n/2+1)$ MMS is 0 for all agents. Thus any allocation satisfies this property.} where the valuations were uniformly distributed from $[0, 1000]$ and then ordered.

\subsection{Individual Thresholds}
For each agent, we utilize binary search to find the largest individual threshold where that agent can form at least $n$ bundles in successful simulations,
as explained in Section \ref{sec:goods-poly}.
To compare the bidirectional and unidirectional bag-filling approaches, we first compute the value that each agent received, as a fraction of his proportional share $Prop_i := v_i(M)/n$. We then plot the minimum ratio any agent received over all 1,000 instances, the average ratio all agents received over 1,000 instances, and the minimum of the average ratios of all agents per instance.

If the valuations were perfectly divisible (say as $m \rightarrow \infty$), we would expect that the 1-of-$(3n/2+1)$ MMS would equate to approximately $\frac{v_{i}(M)}{3n/2} \approx 2/3 \frac{v_{i}(M)}{n} = 2/3 Prop_i$. Since proportionality implies MMS, it is an upper bound for MMS values, each agent's 1-of-$(3n/2+1)$ MMS is at most $2/3 Prop_i$. 

Figure~\ref{fig:individual_thresholds} and Figure~\ref{fig:individual_histograms} show that, while both the bidirectional and unidirectional algorithms exceed the $2/3$ ratio on most instances, the bidirectional algorithm averages slightly higher (about 5\%), with a higher minimum average. Note that the unidirectional bag-filling algorithm is not guaranteed to give every agent 1-of-$(3n/2+1)$ MMS.

\begin{figure}[t]
\begin{center}
\emph{Thresholds computed individually:}
\\
\includegraphics[width=0.45\textwidth]{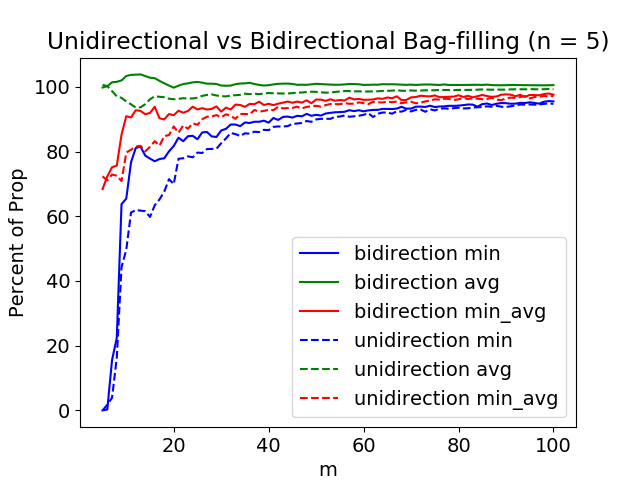}
~~
\includegraphics[width=0.45\textwidth]{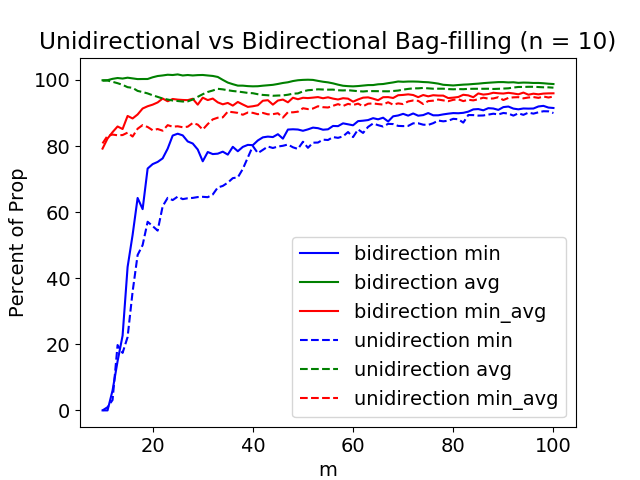}
\\
\includegraphics[width=0.45\textwidth]{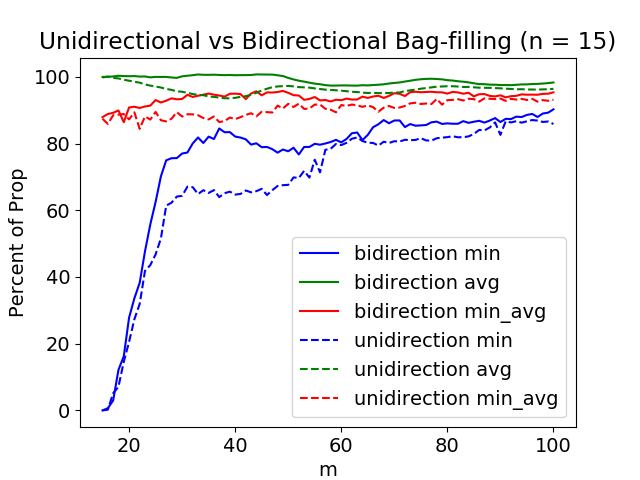}
~~
\includegraphics[width=0.45\textwidth]{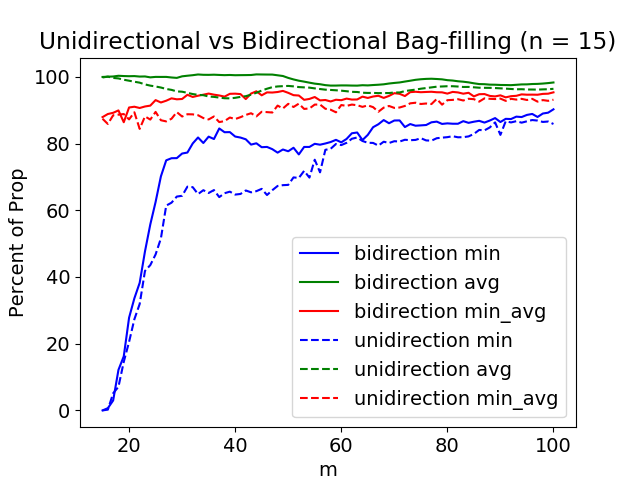}
\end{center}
\caption{
\label{fig:individual_thresholds}
The minimum, average, and minimum average (minimum over all instances of the average value of agents within the instance) value received by agents in both bidirectional and unidirectional bag-filling for various $n$ and $m$. Bidirectional bag-filling always exceeded its unidirectional counterpart.
}
\end{figure}

\begin{figure}
    \centering
    \begin{subfigure}[b]{0.49\textwidth}
        \centering
        \includegraphics[width=\linewidth]{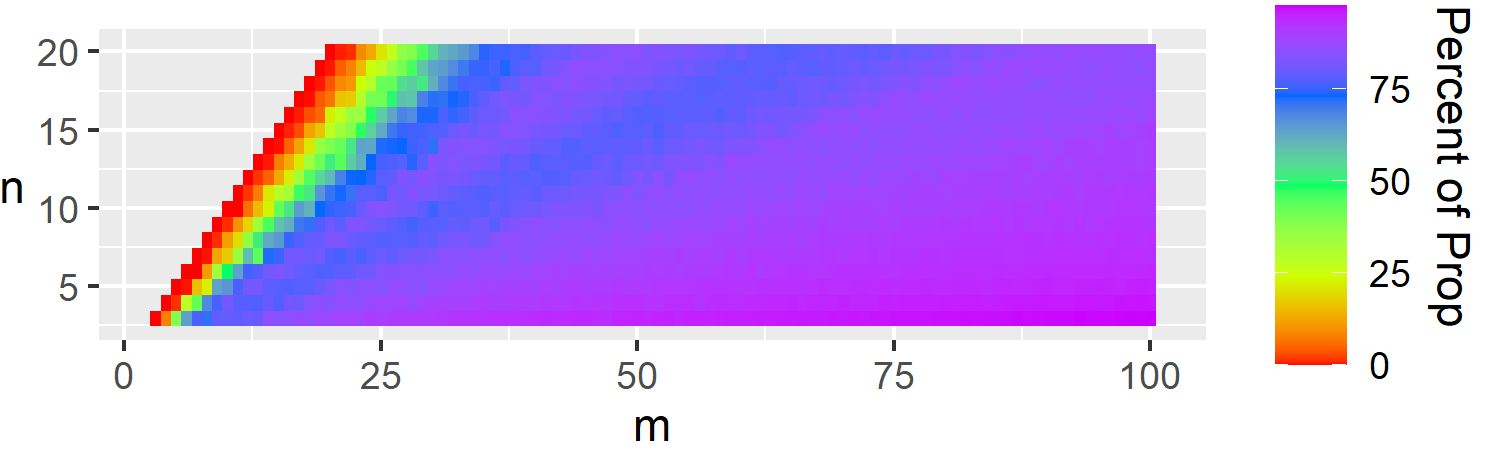}
        \caption{Bidirectional: Minimum}
        \label{fig:y equals x 1}
    \end{subfigure}
    \hfill
    \begin{subfigure}[b]{0.49\textwidth}
        \centering
        \includegraphics[width=\linewidth]{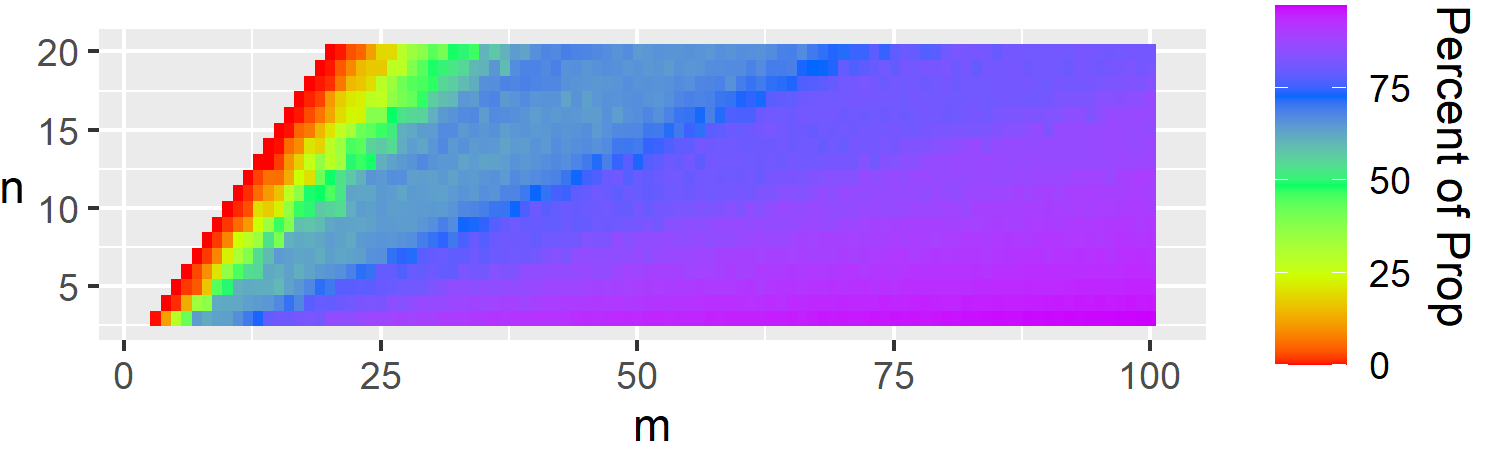}
        \caption{Unidirectional: Minimum}
        \label{fig:three sin x 1}
    \end{subfigure}
    \linebreak
    \begin{subfigure}[b]{0.49\textwidth}
        \centering
        \includegraphics[width=\linewidth]{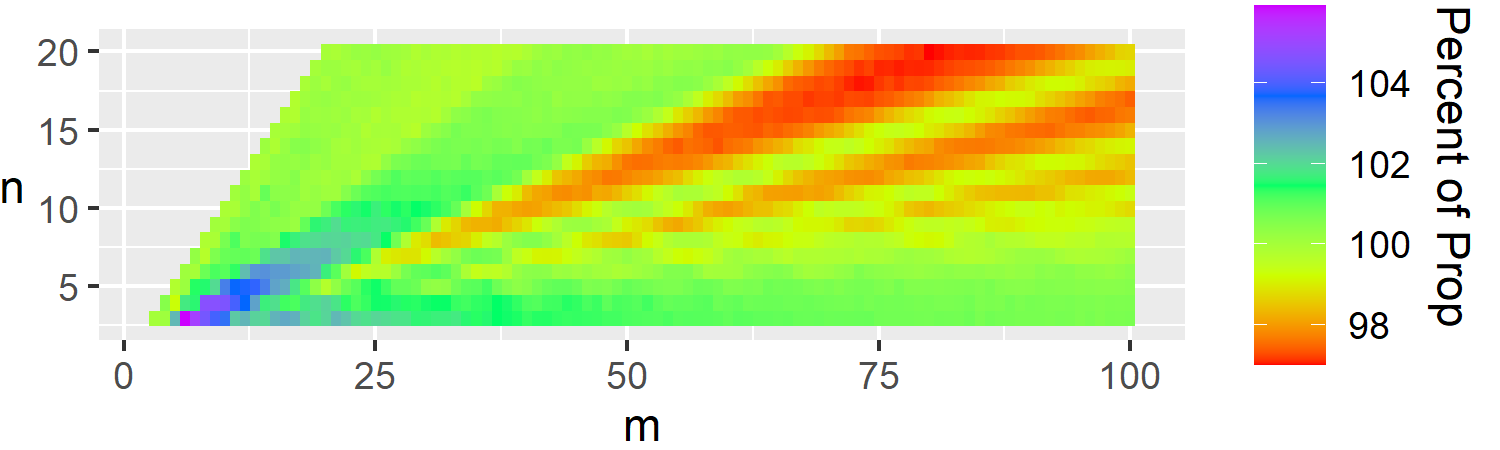}
        \caption{Bidirectional: Average}
        \label{fig:y equals x 2}
    \end{subfigure}
    \hfill
    \begin{subfigure}[b]{0.49\textwidth}
        \centering
        \includegraphics[width=\linewidth]{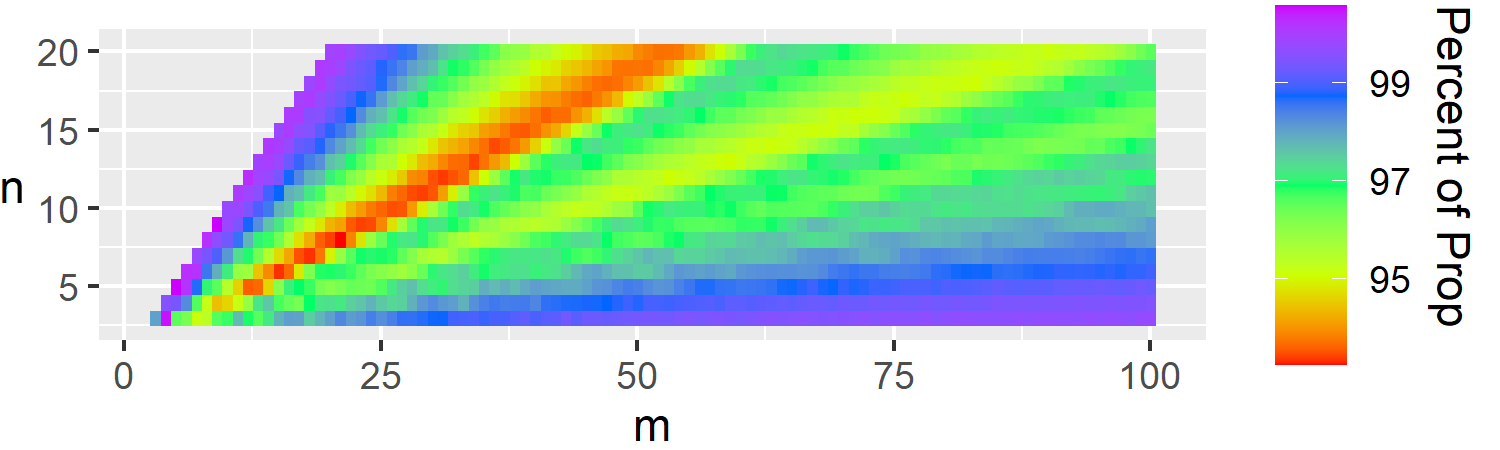} 
        \caption{Unidirectional: Average}
        \label{fig:three sin x 2}
    \end{subfigure}
    \linebreak 
    \begin{subfigure}[b]{0.49\textwidth}
        \centering
        \includegraphics[width=\linewidth]{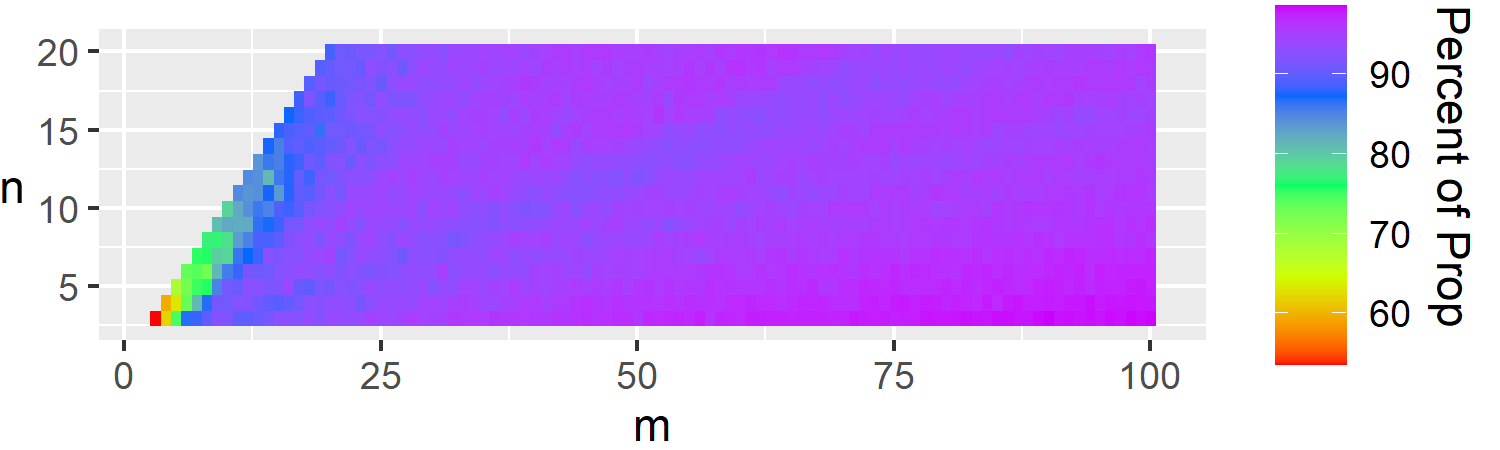}
        \caption{Bidirectional: Minimum Average}
        \label{fig:y equals x 3}
    \end{subfigure}
    \hfill
    \begin{subfigure}[b]{0.49\textwidth}
        \centering
        \includegraphics[width=\linewidth]{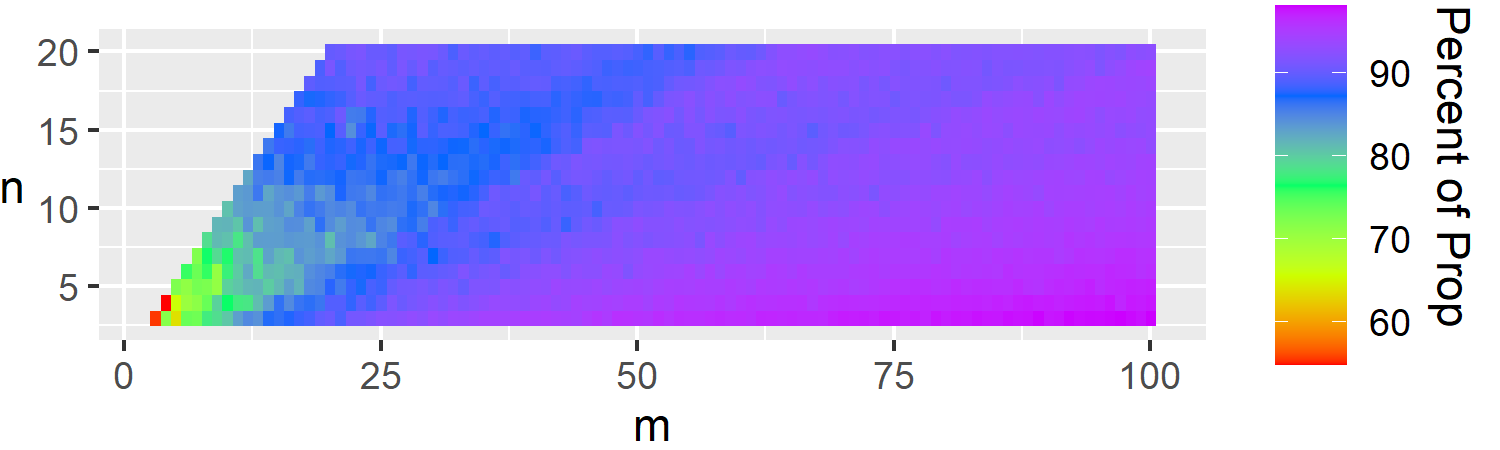} \caption{Unidirectional: Minimum Average}
        \label{fig:three sin x 3}
    \end{subfigure}
    \caption{
     In each figure, the x axis is the number of goods ($m$) and the y axis is the number of agents ($n$). The color represents the fraction of proportional allocations (in percent).
     }
    \label{fig:individual_histograms}
\end{figure}

\clearpage

\subsection{Common Thresholds}
We run binary search to find the largest percentage $t$ where the bidirectional bag-filling algorithm gives each agent that agent at least $t\frac{v_i(M)}{n}$. The ratio was computed up to an error tolerance of $0.1\%$. For this experiment we tracked the minimum ratio and the average ratio guaranteed to all agents across the 1,000 instances.

\begin{figure}
\begin{center}
\emph{Thresholds computed using a common ratio:}
\\
\includegraphics[width=0.45\textwidth]{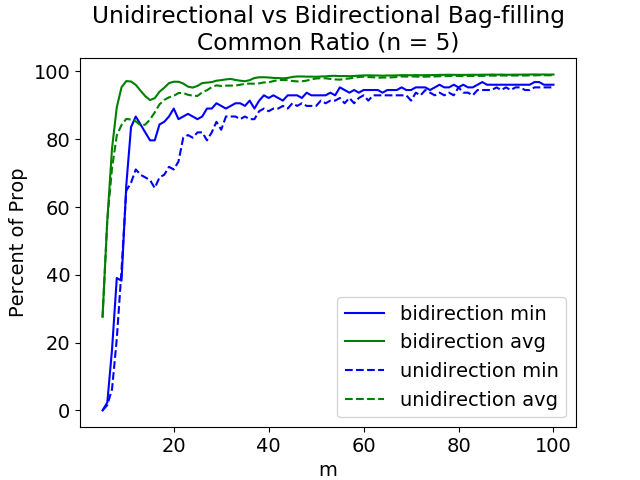}
~~
\includegraphics[width=0.45\textwidth]{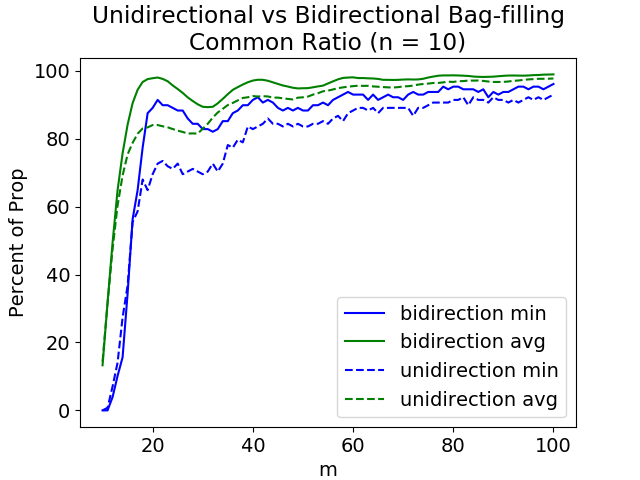}
\\
\includegraphics[width=0.45\textwidth]{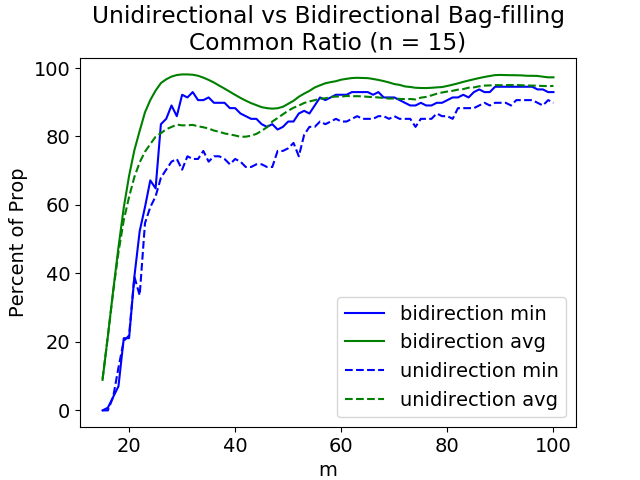}
~~
\includegraphics[width=0.45\textwidth]{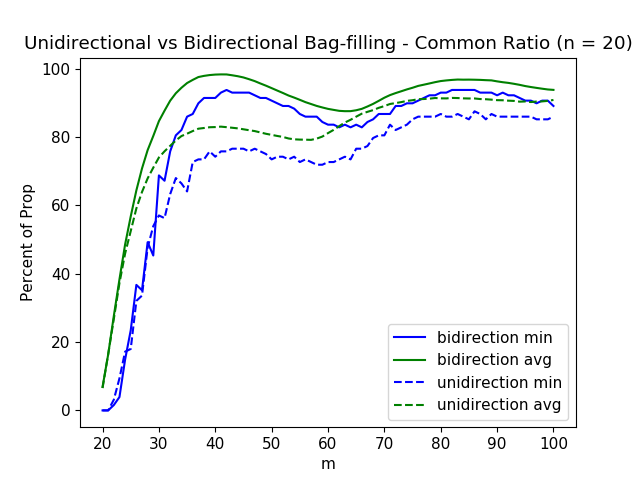}
\end{center}
\caption{
\label{fig:common_thresholds}
The minimum, average, and minimum average (minimum over all instances of the average value of agents within the instance) value received by agents in both bidirectional and unidirectional bag-filling for various $n$ and $m$. Bidirectional bag-filling always exceeded its unidirectional counterpart.
}
\end{figure}

We observe that the bidirectional bag-filling algorithm outperforms the unidirectional bag-filling algorithm, especially when the number of goods is small relative to the number of agents ($m < 2n$).
Interestingly, Figure \ref{fig:common_thresholds} illustrates that bidirectional bag-filling achieves a higher fraction of proportionality with respect to all measures (minimum, average, and minimum average over all instances) compared to the unidirectional bag-filling algorithm.


\begin{figure}
    \centering
    \begin{subfigure}[b]{0.49\textwidth}
        \centering
        \includegraphics[width=\linewidth]{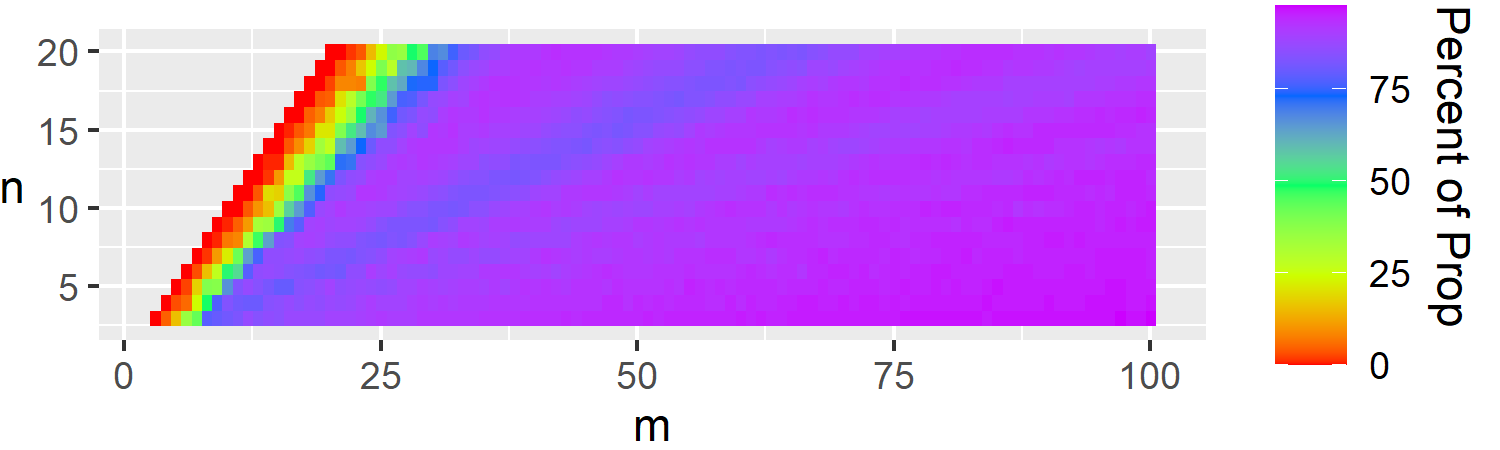}
        \caption{Bidirectional: Minimum}
        \label{fig:y equals x 4}
    \end{subfigure}
    \hfill
    \begin{subfigure}[b]{0.49\textwidth}
        \centering
        \includegraphics[width=\linewidth]{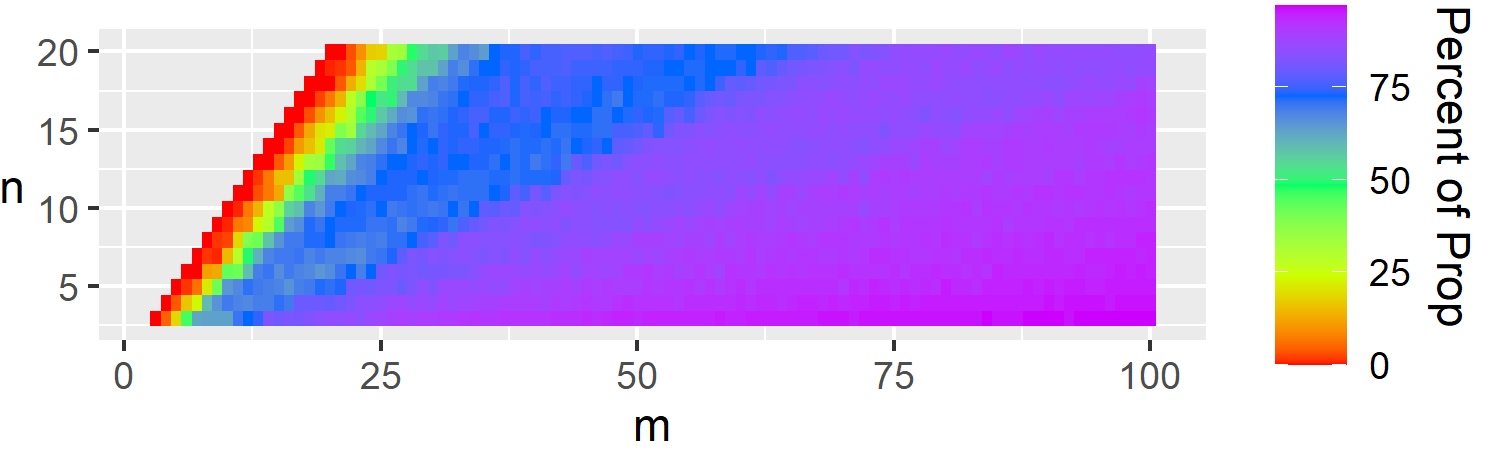}
        \caption{Unidirectional: Minimum}
        \label{fig:three sin x 4}
    \end{subfigure}
    \linebreak 
    \begin{subfigure}[b]{0.49\textwidth}
        \centering
        \includegraphics[width=\linewidth]{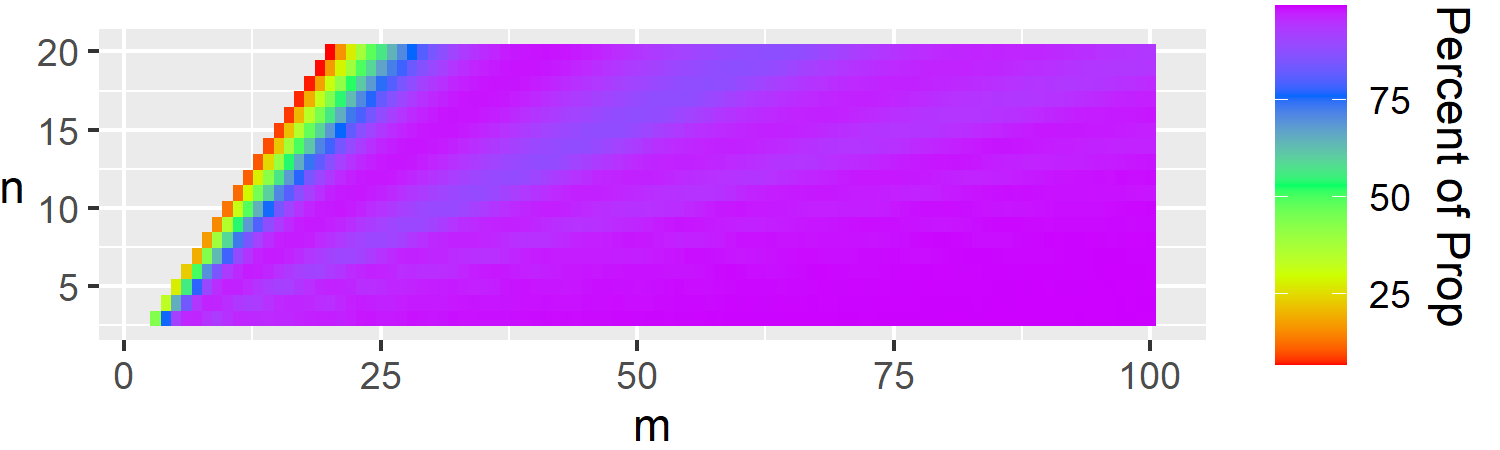}
        \caption{Bidirectional: Average}
        \label{fig:y equals x 5}
    \end{subfigure}
    \hfill
    \begin{subfigure}[b]{0.49\textwidth}
        \centering
        \includegraphics[width=\linewidth]{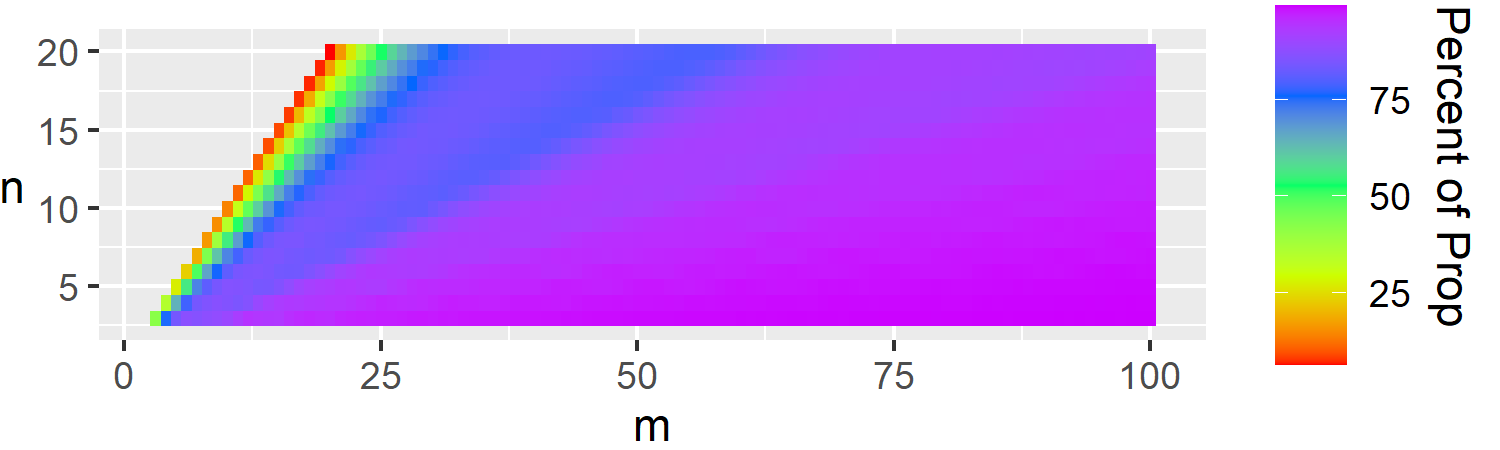} 
        \caption{Unidirectional: Average}
        \label{fig:three sin x 5}
    \end{subfigure}
    \caption{
     In each figure, the x axis is the number of goods ($m$) and the y axis is the number of agents ($n$). The color represents the fraction of proportional allocations (in percent).
     }     
\end{figure}


\newpage
\section{Beyond Additive Valuations} \label{app:responsive}
In this section, we show that no meaningful ordinal MMS approximation is possible when generalizing additive valuations to \emph{responsive} valuations (defined below), even for two agents. Moreover, this result implies that ordinal MMS approximations cannot be extended to \emph{submodular} valuations.

The class of \emph{responsive valuations} was introduced in the literature on matching markets \citep{roth1986allocation,alkan1999properties,klaus2005stable,hatfield2009strategy}
and has recently been used in fair item allocation 
\citep{aziz2019efficient,kyropoulou2020almost,babaioff2021competitive}.
Formally, a valuation $v: 2^M\to \mathbb{R}_+$ is called \emph{responsive} if for any two goods $x, y\in M$ and any subset $Z\subseteq M$:
\begin{align*}
v(x) \leq v(y) \iff v(Z)\leq v(Z\cup \{x\}) \leq v(Z\cup \{y\}).
\end{align*}

Responsive valuations have several equivalent definitions \citep{aziz2015fair}.
One of them uses the notion of \emph{domination}. 
Given a valuation $v$ on individual goods, 
a bundle $X$ is \emph{dominated by} a bundle $Y$, denoted $X\precsim_v Y$, if there is an injection $f: X\to Y$ such that $v(x) \leq v(f(x))$ for all $x\in X$.
The valuation $v$ is called \emph{responsive} 
if 
\begin{align*}
X \precsim_v Y \implies v(X)\leq v(Y).
\end{align*}
Intuitively, responsive valuations presume that agents rank individual goods, and that their ranking of bundles is consistent with the ranking of goods. 

If $v$ is additive then it is clearly responsive, but the opposite is not necessarily true. 
For example, suppose there are four goods with 
$v(w)<v(x)<v(y)<v(z)$.
Then the responsiveness condition does not imply anything about the relation between $\{z\}$ and $\{x,y\}$,
since none of them dominates the other. 
Similarly, responsiveness does not imply anything regarding
$\{w,z\}$ and $\{w,x,y\}$,
since none of them dominates the other. 
So it is possible that $v(\{z\})<v(\{x,y\})<v(\{w,x,y\})<v(\{w,z\})$. This is impossible with additive valuations.

The technique of picking-sequences, used for un-ordering an instance (see Section \ref{sub:ordering}), works for responsive preferences too. In fact, for each agent $i$, the bundle picked by $i$ during the picking-sequence dominates the bundle allocated to $i$ by the algorithm on the ordered instance (the picking-sequence implements the injection $f$).
Similarly, Lemma \ref{lem:bbsfair}, ensuring the existence of a BBFS-fair allocation, holds for responsive preferences too, since the proof only requires containment of bundles. 
Therefore, the class of responsive preferences initially seems like a good candidate for generalizing our results. 

Unfortunately, we show below that, with responsive valuations, no meaningful ordinal approximation is possible, even for two agents. 
This indicates that the ordinal maximin-share approxiation may be too strong for handling non-additive valuations.

\begin{proposition}
\label{prop:responsive}
For any integer $d\geq 1$, there is an instance with two agents with responsive valuations, in which no allocation guarantees both agents their $1$-out-of-$d$ maximin-share.
\end{proposition}

\begin{proof}
We construct an instance with $m = 2^{d^2}-1$ goods. 
They are ranked the same for both agents: $v_i(g_1) > \ldots > v_i(g_m)$ for $i\in\{1,2\}$.

For each $j\in[d^2]$, we denote by $B_j$ the set of bundles that contain a majority of the goods in $\{g_1,\ldots,g_{(2^j)-1}\}$.
Note that $B_j\subseteq 2^M$.
For example:
\begin{itemize}
\item $B_1$ is the set of all bundles that contain $g_1$;
\item $B_2$ is the set of all bundles that contain at least two goods from $\{g_1,g_2,g_3\}$; 
\item $B_3$ is the set of all bundles that contain at least four goods from $\{g_1,\ldots, g_7\}$; 
\item 
$B_{d^2}$ is the set of all bundles that contain at least $2^{d^2-1}$ goods from $M$.
\end{itemize}
Let $X,Y$ be some bundles such that $X\in B_j$ and $Y\in 2^M\setminus B_j$ for some $j\in[d^2]$. The majority assumption implies that $X$ cannot be dominated by $Y$: there cannot be an injection from a majority to a minority. Therefore, a responsive valuation may
assign a larger value to $X$ than to $Y$.

Before proceeding with the proof, we exemplify it for the special case $d=2$.
We define the valuations of two agents as follows.
\begin{itemize}
\item For agent 1, we set $v_1(X) = 1$ for any  $X\in (B_1\cap B_2)\cup (B_3\cap B_4)$, and smaller values for other bundles.
Then $\MMSA{1}{d}_1 = 1$, since $M$ can be partitioned into 
$\{g_1,~ g_2,g_3\}\in B_1\cap B_2$ and $\{g_4,\ldots, g_7,~ g_8,\ldots, g_{15}\}\in B_3\cap B_4$. 
\item For agent 2, 
we set $v_2(X) = 1$ for any $X\in (B_1\cap B_3)\cup (B_2\cap B_4)$, and smaller values for other bundles.
Then $\MMSA{1}{d}_2 = 1$, since $M$ can be partitioned into 
$\{g_1; g_4,\ldots,g_7\}\in B_1\cap B_3$ and $\{g_2,g_3; g_8,\ldots, g_{15}\}\in B_2\cap B_4$.
\end{itemize}
Note that the agents' valuations are consistent with responsiveness. For example, consider a bundle 
$Y\not\in (B_1\cap B_2)\cup (B_3\cap B_4)$.
Then either $Y\not\in B_1\cup B_3$ or $Y\not\in B_1\cup B_4$ or $Y\not\in B_2\cup B_3$ or $Y\not\in B_2\cup B_4$.
In any case, $Y$ cannot dominate any bundle 
$X \in (B_1\cap B_2)\cup (B_3\cap B_4)$.
So assigning to $X$ a higher value than to $Y$ is consistent with responsive valuations.

Suppose now that an allocation $(A_1,A_2)$ gives agent 1 a value of at least $1$. This means that either $A_1\in B_1\cap B_2$ or $A_1\in B_3\cap B_4$. 
If $A_1\in B_1\cap B_2$, then $A_1$ contains $g_1$ and a majority of the goods from $\{g_1,g_2,g_3\}$.
This means that $A_2$ cannot contain $g_1$ and cannot contain a majority of $\{g_1,\ldots,g_3\}$. So $A_2\not\in B_1$ and $A_2 \not \in B_2$. This means that 
$A_2\not\in (B_1\cap B_3)\cup (B_2\cap B_4)$, so the value of agent 2 is less than $\MMSA{1}{d}_2$.
Similarly, if $A_1\in B_3\cap B_4$, then 
$A_2\not\in B_3$ and $A_2 \not \in B_4$, so again 
the value of agent 2 is less than $\MMSA{1}{d}_2$.
We conclude that no allocation gives both agents their 1-out-of-$d$ MMS.

We now generalize this construction to any integer $d$.
\begin{itemize}
\item Let $v_1(X) = 1$ for any bundle satisfying
\begin{align*}
X\in \bigcup_{i=1}^d \left( \bigcap_{j=1}^d B_{(i-1)d+j} \right)
\end{align*}
and smaller values for other bundles.
\item Let $v_2(X) = 1$ for any bundle satisfying
\begin{align*}
X\in \bigcup_{i=1}^d \left( \bigcap_{j=1}^d B_{(j-1)d+i} \right)
\end{align*}
and smaller values for other bundles.
\end{itemize}
Note that agent 1's valuations are consistent with responsiveness, since any bundle not in $\bigcup_{i=1}^d \left( \bigcap_{j=1}^d B_{(i-1)d+j} \right)$ cannot dominate a bundle from $\bigcup_{i=1}^d \left( \bigcap_{j=1}^d B_{(i-1)d+j} \right)$, and simlarly for agent 2.

To compute the agents' MMS, define $d^2$ bundles as follows. For each $j\in[d^2]$, let
\begin{align*}
G_j = \{g_{2^{(j-1)}},\ldots, g_{(2^j)-1}\}.
\end{align*}
For example, $G_1 = \{g_1\}, G_2=\{g_2,g_3\}, G_3 = \{g_4,\ldots,g_7\}$, and so on. 
Note that the $G_j$ are pairwise-disjoint, and for any $j\in[d^2]$, $G_j\in B_j$ since it contains the majority of goods in $\{g_1,\ldots, g_{(2^j)-1}\}$.
Then:
\begin{itemize}
\item $\MMSA{1}{d}_1 = 1$, by the partition 
with parts 
$P_i := \left( \bigcup_{j=1}^d G_{(i-1)d+j} \right)$ for $i\in[d]$.
\item $\MMSA{1}{d}_2 = 1$, by the partition 
with parts 
$Q_i := \left( \bigcup_{j=1}^d G_{(j-1)d+i} \right)$ for $i\in[d]$.
\end{itemize}

Suppose now that an allocation $(A_1,A_2)$ gives agent 1 a value of at least $1$. This means that
$A_1\in \left( \bigcap_{j=1}^d B_{(i_1-1)d+j} \right)$
for some $i_1\in[d]$.
Then, $A_2\not\in B_{(i_1-1)d+j}$ for any $j\in[d]$.
So 
$A_2\not\in \left( \bigcap_{j=1}^d B_{(j-1)d+i_2} \right)$ for any $i_2\in[d]$.
So the value of agent 2 is less than $1$.

We conclude that no allocation gives both agents their 1-out-of-$d$ MMS.
The proof holds for any positive integer $d$.
\end{proof}

\begin{remark}
\citet{babaioff2021competitive} prove that responsive preferences are a subset of submodular preferences. Every submodular preference relation can be represented by a submodular valuation function.
Therefore, the impossibility in Proposition \ref{prop:responsive} extends to submodular valuations too.
\end{remark}

\newpage

\vskip 0.2in
\bibliography{../ref}

\begin{thebibliography}{}

\bibitem[\protect\BCAY{Aigner-Horev\ \BBA\ Segal-Halevi}{Aigner-Horev\ \BBA\
  Segal-Halevi}{2022}]{aigner2022envy}
Aigner-Horev, E.\BBACOMMA\  \BBA\ Segal-Halevi, E. \BBOP2022\BBCP.
\newblock \BBOQ Envy-free matchings in bipartite graphs and their applications
  to fair division\BBCQ\
\newblock {\Bem Information Sciences}, {\Bem 587}, 164--187.

\bibitem[\protect\BCAY{Alkan}{Alkan}{1999}]{alkan1999properties}
Alkan, A. \BBOP1999\BBCP.
\newblock \BBOQ On the properties of stable many-to-many matchings under
  responsive preferences\BBCQ\
\newblock In {\Bem Current Trends in Economics}, \BPGS\ 29--39. Springer.

\bibitem[\protect\BCAY{Amanatidis, Birmpas,\ \BBA\ Markakis}{Amanatidis
  et~al.}{2016}]{amanatidis2016truthful}
Amanatidis, G., Birmpas, G., \BBA\ Markakis, E. \BBOP2016\BBCP.
\newblock \BBOQ On truthful mechanisms for maximin share allocations\BBCQ\
\newblock In {\Bem Proceedings of the International Joint Conference on
  Artificial Intelligence (IJCAI)}, \BPGS\ 31--37.

\bibitem[\protect\BCAY{Amanatidis, Markakis, Nikzad,\ \BBA\ Saberi}{Amanatidis
  et~al.}{2017}]{amanatidis2017approximation}
Amanatidis, G., Markakis, E., Nikzad, A., \BBA\ Saberi, A. \BBOP2017\BBCP.
\newblock \BBOQ Approximation algorithms for computing maximin share
  allocations\BBCQ\
\newblock {\Bem ACM Transactions on Algorithms (TALG)}, {\Bem 13\/}(4), 52.

\bibitem[\protect\BCAY{Assmann, Johnson, Kleitman,\ \BBA\ Leung}{Assmann
  et~al.}{1984}]{assmann1984dual}
Assmann, S.~F., Johnson, D.~S., Kleitman, D.~J., \BBA\ Leung, J.-T.
  \BBOP1984\BBCP.
\newblock \BBOQ On a dual version of the one-dimensional bin packing
  problem\BBCQ\
\newblock {\Bem Journal of algorithms}, {\Bem 5\/}(4), 502--525.

\bibitem[\protect\BCAY{Aziz, Biro, Lang, Lesca,\ \BBA\ Monnot}{Aziz
  et~al.}{2019}]{aziz2019efficient}
Aziz, H., Biro, P., Lang, J., Lesca, J., \BBA\ Monnot, J. \BBOP2019\BBCP.
\newblock \BBOQ Efficient reallocation under additive and responsive
  preferences\BBCQ\
\newblock {\Bem Theoretical Computer Science}, {\Bem 790}, 1--15.

\bibitem[\protect\BCAY{Aziz, Gaspers, Mackenzie,\ \BBA\ Walsh}{Aziz
  et~al.}{2015}]{aziz2015fair}
Aziz, H., Gaspers, S., Mackenzie, S., \BBA\ Walsh, T. \BBOP2015\BBCP.
\newblock \BBOQ Fair assignment of indivisible objects under ordinal
  preferences\BBCQ\
\newblock {\Bem Artificial Intelligence}, {\Bem 227}, 71--92.

\bibitem[\protect\BCAY{Aziz\ \BBA\ Ye}{Aziz\ \BBA\ Ye}{2014}]{aziz2014cake}
Aziz, H.\BBACOMMA\  \BBA\ Ye, C. \BBOP2014\BBCP.
\newblock \BBOQ Cake cutting algorithms for piecewise constant and piecewise
  uniform valuations\BBCQ\
\newblock In {\Bem International Conference on Web and Internet Economics},
  \BPGS\ 1--14. Springer.

\bibitem[\protect\BCAY{Babaioff, Nisan,\ \BBA\ Talgam-Cohen}{Babaioff
  et~al.}{2019}]{babaioff2019fair}
Babaioff, M., Nisan, N., \BBA\ Talgam-Cohen, I. \BBOP2019\BBCP.
\newblock \BBOQ Fair allocation through competitive equilibrium from generic
  incomes\BBCQ\
\newblock In {\Bem Proceedings of the Conference on Fairness, Accountability,
  and Transparency}, \BPGS\ 180--180.

\bibitem[\protect\BCAY{Babaioff, Nisan,\ \BBA\ Talgam-Cohen}{Babaioff
  et~al.}{2021}]{babaioff2021competitive}
Babaioff, M., Nisan, N., \BBA\ Talgam-Cohen, I. \BBOP2021\BBCP.
\newblock \BBOQ Competitive equilibrium with indivisible goods and generic
  budgets\BBCQ\
\newblock {\Bem Mathematics of Operations Research}, {\Bem 46\/}(1), 382--403.

\bibitem[\protect\BCAY{Barman\ \BBA\ Krishna~Murthy}{Barman\ \BBA\
  Krishna~Murthy}{2017}]{barman2017approximation}
Barman, S.\BBACOMMA\  \BBA\ Krishna~Murthy, S.~K. \BBOP2017\BBCP.
\newblock \BBOQ Approximation algorithms for maximin fair division\BBCQ\
\newblock In {\Bem Proceedings of the ACM Conference on Economics and
  Computation}, \BPGS\ 647--664.

\bibitem[\protect\BCAY{Bogomolnaia\ \BBA\ Moulin}{Bogomolnaia\ \BBA\
  Moulin}{2022}]{bogomolnaia2022guarantees}
Bogomolnaia, A.\BBACOMMA\  \BBA\ Moulin, H. \BBOP2022\BBCP.
\newblock \BBOQ Guarantees in fair division: General or monotone
  preferences\BBCQ.
\newblock arXiv preprint 1911.10009.

\bibitem[\protect\BCAY{Bouveret, Chevaleyre,\ \BBA\ Maudet}{Bouveret
  et~al.}{2016}]{Bouveret2016Fair}
Bouveret, S., Chevaleyre, Y., \BBA\ Maudet, N. \BBOP2016\BBCP.
\newblock \BBOQ Fair allocation of indivisible items\BBCQ\
\newblock In Brandt, F., Conitzer, V., Endriss, U., Lang, J., \BBA\ Procaccia,
  A.~D.\BEDS, {\Bem Handbook of Computational Social Choice}, \BCH~13, \BPGS\
  284--310. Cambridge University Press.

\bibitem[\protect\BCAY{Bouveret\ \BBA\ Lema{\^i}tre}{Bouveret\ \BBA\
  Lema{\^i}tre}{2016}]{Bouveret2016}
Bouveret, S.\BBACOMMA\  \BBA\ Lema{\^i}tre, M. \BBOP2016\BBCP.
\newblock \BBOQ Characterizing conflicts in fair division of indivisible goods
  using a scale of criteria\BBCQ\
\newblock {\Bem Autonomous Agents and Multi-Agent Systems}, {\Bem 30\/}(2),
  259--290.

\bibitem[\protect\BCAY{Brams, Kilgour,\ \BBA\ Klamler}{Brams
  et~al.}{2012}]{Brams2012Undercut}
Brams, S.~J., Kilgour, \BBA\ Klamler, C. \BBOP2012\BBCP.
\newblock \BBOQ {The undercut procedure: an algorithm for the envy-free
  division of indivisible items}\BBCQ\
\newblock {\Bem Soc. Choice Welf.}, {\Bem 39\/}(2-3), 615--631.

\bibitem[\protect\BCAY{Budish}{Budish}{2011}]{budish2011combinatorial}
Budish, E. \BBOP2011\BBCP.
\newblock \BBOQ The combinatorial assignment problem: Approximate competitive
  equilibrium from equal incomes\BBCQ\
\newblock {\Bem Journal of Political Economy}, {\Bem 119\/}(6), 1061--1103.

\bibitem[\protect\BCAY{Csirik, Frenk, Labb{\`e},\ \BBA\ Zhang}{Csirik
  et~al.}{1999}]{csirik1999two}
Csirik, J., Frenk, J. B.~G., Labb{\`e}, M., \BBA\ Zhang, S. \BBOP1999\BBCP.
\newblock \BBOQ Two simple algorithms for bin covering\BBCQ\
\newblock {\Bem Acta Cybernetica}, {\Bem 14\/}(1), 13--25.

\bibitem[\protect\BCAY{Csirik, Johnson,\ \BBA\ Kenyon}{Csirik
  et~al.}{2001}]{csirik2001better}
Csirik, J., Johnson, D.~S., \BBA\ Kenyon, C. \BBOP2001\BBCP.
\newblock \BBOQ Better approximation algorithms for bin covering\BBCQ\
\newblock In {\Bem SODA}, \lowercase{\BVOL}~1, \BPGS\ 557--566.

\bibitem[\protect\BCAY{Csirik, Kellerer,\ \BBA\ Woeginger}{Csirik
  et~al.}{1992}]{csirik1992exact}
Csirik, J., Kellerer, H., \BBA\ Woeginger, G. \BBOP1992\BBCP.
\newblock \BBOQ The exact lpt-bound for maximizing the minimum completion
  time\BBCQ\
\newblock {\Bem Operations Research Letters}, {\Bem 11\/}(5), 281--287.

\bibitem[\protect\BCAY{Deuermeyer, Friesen,\ \BBA\ Langston}{Deuermeyer
  et~al.}{1982}]{deuermeyer1982scheduling}
Deuermeyer, B.~L., Friesen, D.~K., \BBA\ Langston, M.~A. \BBOP1982\BBCP.
\newblock \BBOQ Scheduling to maximize the minimum processor finish time in a
  multiprocessor system\BBCQ\
\newblock {\Bem SIAM Journal on Algebraic Discrete Methods}, {\Bem 3\/}(2),
  190--196.

\bibitem[\protect\BCAY{Edmonds\ \BBA\ Pruhs}{Edmonds\ \BBA\
  Pruhs}{2011}]{edmonds2011cake}
Edmonds, J.\BBACOMMA\  \BBA\ Pruhs, K. \BBOP2011\BBCP.
\newblock \BBOQ Cake cutting really is not a piece of cake\BBCQ\
\newblock {\Bem ACM Transactions on Algorithms (TALG)}, {\Bem 7\/}(4), 1--12.

\bibitem[\protect\BCAY{Elkind, Segal-Halevi,\ \BBA\ Suksompong}{Elkind
  et~al.}{2021a}]{ElkindSeSu21c}
Elkind, E., Segal-Halevi, E., \BBA\ Suksompong, W. \BBOP2021a\BBCP.
\newblock \BBOQ Graphical cake cutting via maximin share\BBCQ\
\newblock In {\Bem Proceedings of the International Joint Conference on
  Artificial Intelligence {IJCAI}}.

\bibitem[\protect\BCAY{Elkind, Segal-Halevi,\ \BBA\ Suksompong}{Elkind
  et~al.}{2021b}]{ElkindSeSu21b}
Elkind, E., Segal-Halevi, E., \BBA\ Suksompong, W. \BBOP2021b\BBCP.
\newblock \BBOQ Keep your distance: Land division with separation\BBCQ\
\newblock In {\Bem Proceedings of the International Joint Conference on
  Artificial Intelligence {IJCAI}}.

\bibitem[\protect\BCAY{Elkind, Segal-Halevi,\ \BBA\ Suksompong}{Elkind
  et~al.}{2021c}]{ElkindSeSu21}
Elkind, E., Segal-Halevi, E., \BBA\ Suksompong, W. \BBOP2021c\BBCP.
\newblock \BBOQ Mind the gap: Cake cutting with separation\BBCQ\
\newblock In {\Bem Proceedings of the AAAI Conference on Artificial
  Intelligence}, \BPGS\ 5330--5338.

\bibitem[\protect\BCAY{Frenk\ \BBA\ Kan}{Frenk\ \BBA\
  Kan}{1986}]{frenk1986rate}
Frenk, J. B.~G.\BBACOMMA\  \BBA\ Kan, A.~R. \BBOP1986\BBCP.
\newblock \BBOQ The rate of convergence to optimality of the {LPT} rule\BBCQ\
\newblock {\Bem Discrete Applied Mathematics}, {\Bem 14\/}(2), 187--197.

\bibitem[\protect\BCAY{Garg, McGlaughlin,\ \BBA\ Taki}{Garg
  et~al.}{2018}]{garg2018approximating}
Garg, J., McGlaughlin, P., \BBA\ Taki, S. \BBOP2018\BBCP.
\newblock \BBOQ Approximating maximin share allocations\BBCQ\
\newblock In {\Bem 2nd Symposium on Simplicity in Algorithms (SOSA 2019)}.
  Schloss Dagstuhl-Leibniz-Zentrum fuer Informatik.

\bibitem[\protect\BCAY{Garg\ \BBA\ Taki}{Garg\ \BBA\
  Taki}{2020}]{garg2020improved}
Garg, J.\BBACOMMA\  \BBA\ Taki, S. \BBOP2020\BBCP.
\newblock \BBOQ An improved approximation algorithm for maximin shares\BBCQ\
\newblock In {\Bem Proceedings of the ACM Conference on Economics and
  Computation}, \BPGS\ 379--380.

\bibitem[\protect\BCAY{Ghodsi, HajiAghayi, Seddighin, Seddighin,\ \BBA\
  Yami}{Ghodsi et~al.}{2018}]{ghodsi2018fair}
Ghodsi, M., HajiAghayi, M., Seddighin, M., Seddighin, S., \BBA\ Yami, H.
  \BBOP2018\BBCP.
\newblock \BBOQ Fair allocation of indivisible goods: Improvements and
  generalizations\BBCQ\
\newblock In {\Bem Proceedings of the ACM Conference on Economics and
  Computation}, \BPGS\ 539--556. ACM.

\bibitem[\protect\BCAY{Gourv{\`e}s\ \BBA\ Monnot}{Gourv{\`e}s\ \BBA\
  Monnot}{2019}]{gourves2019maximin}
Gourv{\`e}s, L.\BBACOMMA\  \BBA\ Monnot, J. \BBOP2019\BBCP.
\newblock \BBOQ On maximin share allocations in matroids\BBCQ\
\newblock {\Bem Theoretical Computer Science}, {\Bem 754}, 50--64.

\bibitem[\protect\BCAY{Graham}{Graham}{1966}]{graham1966bounds}
Graham, R.~L. \BBOP1966\BBCP.
\newblock \BBOQ Bounds for certain multiprocessing anomalies\BBCQ\
\newblock {\Bem Bell system technical journal}, {\Bem 45\/}(9), 1563--1581.

\bibitem[\protect\BCAY{Graham}{Graham}{1969}]{graham1969bounds}
Graham, R.~L. \BBOP1969\BBCP.
\newblock \BBOQ Bounds on multiprocessing timing anomalies\BBCQ\
\newblock {\Bem SIAM Journal on Applied Mathematics}, {\Bem 17\/}(2), 416--429.

\bibitem[\protect\BCAY{Halpern\ \BBA\ Shah}{Halpern\ \BBA\
  Shah}{2021}]{halpern2021fair}
Halpern, D.\BBACOMMA\  \BBA\ Shah, N. \BBOP2021\BBCP.
\newblock \BBOQ Fair and efficient resource allocation with partial
  information\BBCQ\
\newblock In {\Bem Proceedings of the International Joint Conference on
  Artificial Intelligence {IJCAI}}, \BPGS\ 224--230.

\bibitem[\protect\BCAY{Hatfield}{Hatfield}{2009}]{hatfield2009strategy}
Hatfield, J.~W. \BBOP2009\BBCP.
\newblock \BBOQ Strategy-proof, efficient, and nonbossy quota allocations\BBCQ\
\newblock {\Bem Social Choice and Welfare}, {\Bem 33\/}(3), 505--515.

\bibitem[\protect\BCAY{Herreiner\ \BBA\ Puppe}{Herreiner\ \BBA\
  Puppe}{2002}]{herreiner2002simple}
Herreiner, D.\BBACOMMA\  \BBA\ Puppe, C. \BBOP2002\BBCP.
\newblock \BBOQ {A simple procedure for finding equitable allocations of
  indivisible goods}\BBCQ\
\newblock {\Bem Soc. Choice Welf.}, {\Bem 19\/}(2), 415--430.

\bibitem[\protect\BCAY{Hosseini\ \BBA\ Searns}{Hosseini\ \BBA\
  Searns}{2021}]{hosseini2021mms}
Hosseini, H.\BBACOMMA\  \BBA\ Searns, A. \BBOP2021\BBCP.
\newblock \BBOQ Guaranteeing maximin shares: Some agents left behind\BBCQ\
\newblock In Zhou, Z.-H.\BED, {\Bem Proceedings of the International Joint
  Conference on Artificial Intelligence {IJCAI}}, \BPGS\ 238--244.
  International Joint Conferences on Artificial Intelligence Organization.
\newblock Main Track.

\bibitem[\protect\BCAY{Hosseini, Searns,\ \BBA\ Segal-Halevi}{Hosseini
  et~al.}{2022}]{hosseini2022ordinal}
Hosseini, H., Searns, A., \BBA\ Segal-Halevi, E. \BBOP2022\BBCP.
\newblock \BBOQ Ordinal maximin share approximation for chores\BBCQ\
\newblock In {\Bem Proceedings of the International Conference on Autonomous
  Agents and MultiAgent Systems}, \BPG\ forthcoming.

\bibitem[\protect\BCAY{Huang\ \BBA\ Lu}{Huang\ \BBA\
  Lu}{2021}]{huang2021algorithmic}
Huang, X.\BBACOMMA\  \BBA\ Lu, P. \BBOP2021\BBCP.
\newblock \BBOQ An algorithmic framework for approximating maximin share
  allocation of chores\BBCQ\
\newblock In {\Bem Proceedings of the ACM Conference on Economics and
  Computation (EC-21)}, \BPGS\ 630--631.

\bibitem[\protect\BCAY{Jansen\ \BBA\ Solis-Oba}{Jansen\ \BBA\
  Solis-Oba}{2003}]{jansen2003asymptotic}
Jansen, K.\BBACOMMA\  \BBA\ Solis-Oba, R. \BBOP2003\BBCP.
\newblock \BBOQ An asymptotic fully polynomial time approximation scheme for
  bin covering\BBCQ\
\newblock {\Bem Theoretical Computer Science}, {\Bem 306\/}(1-3), 543--551.

\bibitem[\protect\BCAY{Johnson}{Johnson}{1973}]{johnson1973near}
Johnson, D.~S. \BBOP1973\BBCP.
\newblock {\Bem Near-optimal bin packing algorithms}.
\newblock Ph.D.\ thesis, Massachusetts Institute of Technology.

\bibitem[\protect\BCAY{Klaus\ \BBA\ Klijn}{Klaus\ \BBA\
  Klijn}{2005}]{klaus2005stable}
Klaus, B.\BBACOMMA\  \BBA\ Klijn, F. \BBOP2005\BBCP.
\newblock \BBOQ Stable matchings and preferences of couples\BBCQ\
\newblock {\Bem Journal of Economic Theory}, {\Bem 121\/}(1), 75--106.

\bibitem[\protect\BCAY{Kuhn}{Kuhn}{1967}]{kuhn1967games}
Kuhn, H.~W. \BBOP1967\BBCP.
\newblock {\Bem {On games of fair division}}, \BPGS\ 29--37.
\newblock Princeton University Press.

\bibitem[\protect\BCAY{Kurokawa, Procaccia,\ \BBA\ Wang}{Kurokawa
  et~al.}{2018}]{kurokawa2018fair}
Kurokawa, D., Procaccia, A.~D., \BBA\ Wang, J. \BBOP2018\BBCP.
\newblock \BBOQ Fair enough: Guaranteeing approximate maximin shares\BBCQ\
\newblock {\Bem Journal of the ACM (JACM)}, {\Bem 65\/}(2), 8.

\bibitem[\protect\BCAY{Kyropoulou, Suksompong,\ \BBA\ Voudouris}{Kyropoulou
  et~al.}{2020}]{kyropoulou2020almost}
Kyropoulou, M., Suksompong, W., \BBA\ Voudouris, A.~A. \BBOP2020\BBCP.
\newblock \BBOQ Almost envy-freeness in group resource allocation\BBCQ\
\newblock {\Bem Theoretical Computer Science}, {\Bem 841}, 110--123.

\bibitem[\protect\BCAY{Lang\ \BBA\ Rothe}{Lang\ \BBA\
  Rothe}{2016}]{lang2016fair}
Lang, J.\BBACOMMA\  \BBA\ Rothe, J. \BBOP2016\BBCP.
\newblock \BBOQ Fair division of indivisible goods\BBCQ\
\newblock In {\Bem Economics and Computation}, \BPGS\ 493--550. Springer.

\bibitem[\protect\BCAY{Lipton, Markakis, Mossel,\ \BBA\ Saberi}{Lipton
  et~al.}{2004}]{lipton2004approximately}
Lipton, R.~J., Markakis, E., Mossel, E., \BBA\ Saberi, A. \BBOP2004\BBCP.
\newblock \BBOQ On approximately fair allocations of indivisible goods\BBCQ\
\newblock In {\Bem Proceedings of the ACM conference on Electronic commerce},
  \BPGS\ 125--131. ACM.

\bibitem[\protect\BCAY{Markakis}{Markakis}{2017}]{markakis2017approximation}
Markakis, E. \BBOP2017\BBCP.
\newblock \BBOQ Approximation algorithms and hardness results for fair division
  with indivisible goods\BBCQ\
\newblock {\Bem Trends in Computational Social Choice}, 231--247.

\bibitem[\protect\BCAY{McGlaughlin\ \BBA\ Garg}{McGlaughlin\ \BBA\
  Garg}{2020}]{mcglaughlin2020improving}
McGlaughlin, P.\BBACOMMA\  \BBA\ Garg, J. \BBOP2020\BBCP.
\newblock \BBOQ Improving {Nash} social welfare approximations\BBCQ\
\newblock {\Bem Journal of Artificial Intelligence Research}, {\Bem 68},
  225--245.

\bibitem[\protect\BCAY{Menon\ \BBA\ Larson}{Menon\ \BBA\
  Larson}{2020}]{menon2020algorithmic}
Menon, V.\BBACOMMA\  \BBA\ Larson, K. \BBOP2020\BBCP.
\newblock \BBOQ Algorithmic stability in fair allocation of indivisible goods
  among two agents\BBCQ.
\newblock arXiv preprint 2007.15203.

\bibitem[\protect\BCAY{Mertens}{Mertens}{2001}]{mertens2001physicist}
Mertens, S. \BBOP2001\BBCP.
\newblock \BBOQ A physicist's approach to number partitioning\BBCQ\
\newblock {\Bem Theoretical Computer Science}, {\Bem 265\/}(1-2), 79--108.

\bibitem[\protect\BCAY{Moulin}{Moulin}{1990}]{Moulin1990Uniform}
Moulin, H. \BBOP1990\BBCP.
\newblock \BBOQ {Uniform externalities: Two axioms for fair allocation}\BBCQ\
\newblock {\Bem Journal of Public Economics}, {\Bem 43\/}(3), 305--326.

\bibitem[\protect\BCAY{Moulin}{Moulin}{1992}]{Moulin1992Welfare}
Moulin, H. \BBOP1992\BBCP.
\newblock \BBOQ {Welfare bounds in the cooperative production problem}\BBCQ\
\newblock {\Bem Games and Economic Behavior}, {\Bem 4\/}(3), 373--401.

\bibitem[\protect\BCAY{Moulin}{Moulin}{2019}]{doi:10.1146/annurev-economics-080218-025559}
Moulin, H. \BBOP2019\BBCP.
\newblock \BBOQ Fair division in the internet age\BBCQ\
\newblock {\Bem Annual Review of Economics}, {\Bem 11\/}(1), 407--441.

\bibitem[\protect\BCAY{Nguyen, Nguyen,\ \BBA\ Rothe}{Nguyen
  et~al.}{2017}]{nguyen2017approximate}
Nguyen, N.-T., Nguyen, T.~T., \BBA\ Rothe, J. \BBOP2017\BBCP.
\newblock \BBOQ Approximate solutions to max-min fair and proportionally fair
  allocations of indivisible goods\BBCQ\
\newblock In {\Bem Proceedings of the Conference on Autonomous Agents and
  MultiAgent Systems}, \BPGS\ 262--271. International Foundation for Autonomous
  Agents and Multiagent Systems.

\bibitem[\protect\BCAY{Procaccia\ \BBA\ Wang}{Procaccia\ \BBA\
  Wang}{2014}]{procaccia2014fair}
Procaccia, A.~D.\BBACOMMA\  \BBA\ Wang, J. \BBOP2014\BBCP.
\newblock \BBOQ Fair enough: Guaranteeing approximate maximin shares\BBCQ\
\newblock In {\Bem Proceedings of the ACM conference on Economics and
  computation}, \BPGS\ 675--692. ACM.

\bibitem[\protect\BCAY{Roth}{Roth}{1986}]{roth1986allocation}
Roth, A.~E. \BBOP1986\BBCP.
\newblock \BBOQ On the allocation of residents to rural hospitals: A general
  property of two-sided matching markets\BBCQ\
\newblock {\Bem Econometrica}, {\Bem 54\/}(2), 425--427.

\bibitem[\protect\BCAY{Searns\ \BBA\ Hosseini}{Searns\ \BBA\
  Hosseini}{2020}]{Searns_Hosseini_2020}
Searns, A.\BBACOMMA\  \BBA\ Hosseini, H. \BBOP2020\BBCP.
\newblock \BBOQ Fairness does not imply satisfaction (student abstract)\BBCQ\
\newblock {\Bem Proceedings of the AAAI Conference on Artificial Intelligence},
  {\Bem 34\/}(10), 13911--13912.

\bibitem[\protect\BCAY{Segal-Halevi}{Segal-Halevi}{2019}]{segal2019maximin}
Segal-Halevi, E. \BBOP2019\BBCP.
\newblock \BBOQ The maximin share dominance relation\BBCQ.
\newblock arXiv preprint 1912.08763.

\bibitem[\protect\BCAY{Segal-Halevi}{Segal-Halevi}{2020}]{segal2020competitive}
Segal-Halevi, E. \BBOP2020\BBCP.
\newblock \BBOQ Competitive equilibrium for almost all incomes: existence and
  fairness\BBCQ\
\newblock {\Bem Autonomous Agents and Multi-Agent Systems}, {\Bem 34\/}(1),
  1--50.

\bibitem[\protect\BCAY{Truszczynski\ \BBA\ Lonc}{Truszczynski\ \BBA\
  Lonc}{2020}]{truszczynski2020maximin}
Truszczynski, M.\BBACOMMA\  \BBA\ Lonc, Z. \BBOP2020\BBCP.
\newblock \BBOQ Maximin share allocations on cycles\BBCQ\
\newblock {\Bem Journal of Artificial Intelligence Research}, {\Bem 69},
  613--655.

\bibitem[\protect\BCAY{Woeginger}{Woeginger}{1997}]{woeginger1997polynomial}
Woeginger, G.~J. \BBOP1997\BBCP.
\newblock \BBOQ A polynomial-time approximation scheme for maximizing the
  minimum machine completion time\BBCQ\
\newblock {\Bem Operations Research Letters}, {\Bem 20\/}(4), 149--154.

\end{thebibliography}
\bibliographystyle{theapa}

\end{document}